\documentclass[aps,twocolumn,preprintnumbers,groupedaddress]{revtex4-2}
\usepackage[T1]{fontenc}
\usepackage[margin=0.8in]{geometry} 
\usepackage{amsmath,amssymb,amsthm,mathtools}
\usepackage{dsfont}
\usepackage[svgnames]{xcolor}
\usepackage[
  colorlinks=true,
  linkcolor=DarkRed,
  citecolor=DarkGreen,
  urlcolor=DarkBlue
]{hyperref}

\usepackage{xurl}
\usepackage[
  colorlinks=true,
  linkcolor=DarkRed,
  citecolor=DarkGreen,
  urlcolor=DarkBlue
]{hyperref}

\usepackage{mathrsfs}
\usepackage{enumitem}
\usepackage{cleveref}

\usepackage{physics}
\usepackage{fancyhdr}
\usepackage{comment}
\newtheorem{definition}{Definition}
\newtheorem{theorem}{Theorem}
\newtheorem{lemma}{Lemma}
\newtheorem{corollary}{Corollary}
\newtheorem*{remark}{Remark}

\newcommand{\e}[1]{\mathrm{e}^{#1}}

\newcommand{\im}{\mathrm{i}}
\usepackage[font=small,labelfont=bf]{caption}
\usepackage{multirow}
\usepackage{subcaption}

\newcommand{\bkm}[2]{\langle #1, #2\rangle^{\mathrm{BKM}}_\rho}

\usepackage{amsfonts, amsmath, amssymb}
\usepackage{xcolor}
\usepackage{tikz}
\usepackage{wrapfig}
\usetikzlibrary{arrows.meta,positioning}
\usepackage{titlesec}

\usepackage{hyperref}
\hypersetup{pdfborder={0 0 0}}

\titlespacing{\subsection}{5pt}{5pt}{5pt}
\titlespacing{\section}{2pt}{6pt}{6pt}
\titlespacing{\subsubsection}{5pt}{6pt}{6pt}

\begin{document}

\preprint{LA-UR-26-23901}
\title{\textbf{Operator Score Matching for Learning Quantum Hamiltonians}}

\author{Shreya Shukla}
\email{sshukla@lanl.gov}
\affiliation{Theoretical Division, Los Alamos National Laboratory, Los Alamos, New Mexico 87545, USA}
\affiliation{Center for Quantum Computing, Los Alamos National Laboratory, Los Alamos, New Mexico 87545, USA}

\author{Abhijith Jayakumar}
\email{abhijithj@lanl.gov}
\affiliation{Theoretical Division, Los Alamos National Laboratory, Los Alamos, New Mexico 87545, USA}
\affiliation{Center for Quantum Computing, Los Alamos National Laboratory, Los Alamos, New Mexico 87545, USA}

\author{Andrey Y.\ Lokhov}
\email{lokhov@lanl.gov}
\affiliation{Theoretical Division, Los Alamos National Laboratory, Los Alamos, New Mexico 87545, USA}
\affiliation{Center for Quantum Computing, Los Alamos National Laboratory, Los Alamos, New Mexico 87545, USA}

\begin{abstract}
    Learning quantum Hamiltonians from low-temperature thermal state measurements is a fundamental problem in quantum physics. The scalability of existing methods is limited by the complexity of semidefinite optimization problems or partition function computation. Here, we develop a quantum analog of classical score matching that exploits generalized notions of derivatives and integration by parts in operator algebras. We introduce a new \emph{Operator Score Matching} loss function that recovers Hamiltonian parameters using simple gradient descent without the need to compute intractable normalization constants.  
    Numerical experiments on the Gibbs states of the transverse-field Ising model, XXZ spin chain,
    Fermi-Hubbard model, and lattice $\phi^4$ theory demonstrate efficient parameter recovery using
    finite-depth truncations of nested commutator expansions.
    Our results establish Operator Score Matching as a practical, partition-function-free framework for quantum Hamiltonian learning.
\end{abstract}

\maketitle

\section{INTRODUCTION}

Determining the microscopic Hamiltonian governing a quantum system from experimental observations is among the most fundamental inverse problems in physics. While forward simulation (predicting observables given a Hamiltonian) has seen dramatic advances through tensor networks~\cite{orus2014,cirac2021}, quantum Monte Carlo~\cite{foulkes2001}, and neural network quantum states~\cite{carleo2017,stokes2020,choo2020}, the inverse problem of \emph{learning} the Hamiltonian remains substantially more challenging. This challenge is central to quantum computing, where validating engineered Hamiltonians against target models is essential~\cite{Cirac2012,Bloch2012}, and to condensed matter physics, where extracting effective low-energy descriptions from spectroscopic or scattering data could reveal emergent phenomena~\cite{keimer2017}.

The specific setting we consider, learning the parameters of a local Hamiltonian from copies of its Gibbs state $\rho \propto \e{-H}$ (where we absorb the inverse temperature $\beta$ into $H$), has a rich recent history~\cite{Cra+10,SLP11,AA23}. Anshu et al.~\cite{anshu2021} first established polynomial sample-complexity bounds for this task, but their algorithm has a worst-case runtime exponential in the system size, requiring a computation equivalent in difficulty to evaluating the partition function, which is intractable even classically~\cite{Mon15}. Haah et al.~\cite{HKT22} gave an efficient algorithm for the high-temperature regime and proved that this scaling is optimal up to constants in the exponent. The low-temperature regime (precisely where quantum phenomena such as topological order and superconductivity are most prominent~\cite{AFOV08}), however, remained open until Bakshi et al.~\cite{bakshi2023learningquantumhamiltonianstemperature} (BLMT) gave a polynomial-time algorithm valid at any constant temperature. Their approach introduces a flat polynomial approximation to the exponential, formulates learning as a polynomial constraint system, and solves it via the Sum-of-Squares (SoS) hierarchy~\cite{barak_sos,FKP19}.

Although theoretically efficient for fixed inverse temperature, the method of Bakshi et al. relies on a high-degree sum-of-squares relaxation. The practical scalability of such relaxations is often limited by the size and computational cost of the resulting semidefinite programs, particularly as the number of variables and polynomial degree increase \cite{ahmadi2019dsos, majumdar2020scalability}. Since the degree required by the algorithm grows rapidly with inverse temperature and the target accuracy, the resulting optimization becomes computationally demanding for larger systems, especially in the low-temperature regime. This motivates learning methods that retain the advantage of avoiding the partition function while employing a simpler optimization procedure.


A natural source of inspiration for circumventing the partition function bottleneck comes from classical statistics. Score matching~\cite{hyvarinen2005} avoids the computation of intractable normalization factors by matching the gradient of the log-probability, or the score function, rather than the probability itself, yielding consistent estimators that never require computing the partition function. The idea has found wide application in machine learning~\cite{song2019,song2021}, energy-based models~\cite{lecun2006}, and statistical physics, where it enables learning of renormalization group flows from Monte Carlo data~\cite{Shukla2025_LSFT,Shanahan:2018vcv}. Very recently, rigorous finite-sample bounds have been established for learning classical Gibbs distributions with continuous variables using score matching~\cite{smedira2026finite}. 

Extending these ideas to quantum systems, however, encounters a fundamental obstacle: the intrinsic non-commutativity of quantum mechanics. Consequently, the classical score $\nabla_x \log p(x)$ lacks a unique quantum analog, as different operator orderings and notions of differentiation yield inequivalent generalizations. Furthermore, naive quantum generalizations of classical loss functions fail to satisfy two essential criteria: \emph{consistency} (possessing a global minimum exactly at the true model parameters) and \emph{estimability} (being directly computable from copies of the target density matrix without prior knowledge of the true parameters).

In this work, we develop an algorithm building on the philosophy of classical score matching while introducing genuinely quantum ingredients. We define an operator score based on conjugation by the density matrix, which eliminates the dependence of the loss on the partition function. 
Using this operator score, we define a loss that can be estimated using independent copies of the unknown density matrix by measuring expectation values of local operators. Moreover, we can show that this loss has a unique global minimum that corresponds to the unknown Hamiltonian.  The method also applies to any system with a finite-dimensional Hilbert space, including systems beyond qubits.  
Interestingly, while the classical score $\nabla\log p(x) = -\nabla H(x)$ is a sum of local terms (since differentiation commutes with the Hamiltonian), the quantum score operator generates non-local contributions through nested commutators, as depicted schematically in Fig.~\ref{fig:learning_in_detail}.

For this new operator-score loss, we establish a local Polyak--\L{}ojasiewicz (PL) inequality~\cite{polyak1963,karimi2016} that guarantees exponential convergence of gradient descent near the true parameters. The PL constant $\mu$ is controlled by the minimum eigenvalue of a Gram matrix $G(\theta)$ that governs the convergence rate and inverse temperature $\beta$. 

We validate the framework numerically on four models spanning different physical settings: the transverse-field Ising model, the XXZ spin chain, the Fermi--Hubbard model, and lattice $\phi^4$ theory. The fermionic example uses Majorana fermion direction operators and works directly in Fock space, confirming that the method requires no mapping to qubits and is agnostic to particle statistics. The framework extends naturally to lattice field theories, where score matching has already proven valuable in the classical setting~\cite{Shukla2025_LSFT}, opening possibilities for data-driven approaches to strongly correlated quantum systems.\\
\\


\section{LEARNING PROBLEM AND OPERATOR SCORE MATCHING}

\subsection{Problem formulation}

We consider the parametrized Hamiltonian
\begin{equation}
H_\theta=\sum_{i=1}^m \theta_i h_i,
\qquad \theta\in\mathbb{R}^m,
\label{eq:Htheta_main}
\end{equation}
with fixed Hermitian operators $\{h_i\}$ and a data state $\rho \propto \e{-H'}$. At the true parameters $\theta^*$, $H_{\theta^\ast} = H'$. We work in the standard setting \cite{anshu2021,bakshi2023learningquantumhamiltonianstemperature} where the operator basis $\{h_i\}$ for representing these Hamiltonians is known, but the associated parameters are not. The goal is to recover $\theta^* $ such that $H_{\theta^*} = H'$ from expectations measured in the state $\rho$.

\subsection{Operator Score Matching}

The key insight for attacking the problem is that quantum systems lend themselves to derivations that serve as analogs of derivative operators in classical theories. We can circumvent the fundamental bottleneck of computing the partition function using these derivations and then match these generalized gradients between the parametrized and data Hamiltonians to solve the learning problem.\\
\\
\textbf{Classical score matching.} Let us recall the classical setting that motivates our approach. For a probability
density $p_\theta(x) \propto \e{-E_\theta(x)}$ (we set $\beta =1$, unless otherwise specified), the 
\emph{score} is the gradient of the log-density,
$s_\theta(x) := \nabla_x \log p_\theta(x) = -\nabla_x E_\theta(x)$, which is
manifestly independent of the partition function $Z(\theta)$. The idea behind score matching is therefore to 
estimate the parameters $\theta$ by minimizing the following loss
\begin{align}
    \mathcal{L}(\theta) &= \mathbb{E}_{p}\,\|s_\theta - s_{\mathrm{data}}\|^2\;.
\end{align}

Minimizing this loss over the parameters $\theta$, in principle, allows us to learn the unknown distribution $p$ as this loss vanishes only when $p_\theta = p$.
However, as written, this loss depends on the unknown data score $s_{\mathrm{data}}$. This issue can be avoided by using an integration-by-parts identity that allows us to express the loss in a form that can be optimized using only samples of the distribution~\cite{hyvarinen2005}. In the classical setting, this framework can be extended to write down more general estimators \cite{shukla2025_EFL}.\\
\\
\textbf{Defining the quantum score and the loss.} To define a consistent and estimable quantum analog that works with density matrices, we first replace the classical gradient $\nabla_x$
by a \emph{derivative operator} built from commutators, so that the resulting
score is again independent of $Z(\theta)$ and the dependence on the unknown
target is removed by an analogous trace integration-by-parts identity (see Eq.~\eqref{eq:IBP_main}).



To construct the operator score loss, we define the generalized derivative of an operator $A$ along the direction $F_k$ as
\begin{align}
    \mathscr D_k A &:= 
    [F_k,e^{-A}]e^A .
    \label{eq:score_from_RU}
\end{align}

From this, we can compute the \emph{quantum score} of the model density
$p_\theta  = e^{-H_\theta}/Z(\theta)$ (along the direction $F_k$) as
\begin{align}
    \mathscr D_k(-\log p_\theta)
    &:=
    [F_k,e^{\log p_\theta}]e^{-\log p_\theta},
    \notag \\
    &=
    [F_k,e^{-H_\theta}Z(\theta)^{-1}]e^{H_\theta}Z(\theta),
    \notag \\
    &=
    [F_k,e^{-H_\theta}]e^{H_\theta},
    \label{eq:quantscore}
\end{align}
Using this notion, we can formally define the operator score-matching loss:
\begin{align}
    \mathcal L_{\rm osm}(\theta)
    =
    \sum_k
    \left\langle
    |\mathscr D_kH_\theta-\mathscr D_kH'|^2
    \right\rangle_\rho .
    \label{eq:Lstd_def}
\end{align}

It is clear that $\mathcal{L}(\theta^*) = 0$; hence the true Hamiltonian is a global minimizer of the loss. However, this does not imply that it is unique. We can show that by choosing an appropriate set of ${F_k}$ operators, we can always write down a loss function with a unique global minimum at the true Hamiltonian parameters.

Since $\rho$ is a full-rank positive matrix, the condition $\mathcal L (\theta) = 0$ implies that $\mathscr{D}_k H_\theta = \mathscr{D}_k H'$ for all directions $k$. Expanding the derivative and rearranging gives
\begin{align}
    \e{-H_\theta}F_k\,\e{H_\theta} &= \e{-H'}F_k\,\e{H'}\quad \forall k\notag \\
    \implies [ \e{H'}\e{-H_\theta}, F_k ] &= 0 \quad \forall k\;.
\end{align}
Thus, a vanishing loss implies that $Y := \e{H'}\e{-H_\theta}$ commutes with all direction operators. Clearly, this also implies that $Y$ commutes with any linear combination of $F_k$ operators. Moreover, since $[Y, F_k F_{k'}] = [Y,F_k]F_{k'} + ~F_k[Y, F_{k'}],$ this also implies that the commutator vanishes for the entire matrix algebra generated by the $F_k$ operators. 

Thus, choosing the direction operators $F_k$ to generate the full algebra forces $Y = c\cdot\mathds{1}$, which implies $H' = H_\theta + \log c\mathds{1}$, i.e., the Hamiltonian is uniquely fixed by $\mathcal{L}(\theta) = 0$, up to an unobservable energy shift.

This also shows that the number of $F_k$ operators we need to choose is much smaller than the dimension of the Hilbert space. For instance, the full algebra of an $n$-qubit system can be generated using only single-qubit (Pauli) $X \equiv \sigma^x$ and $Z \equiv \sigma^z$ operators. The choice of $F_k$ operators also affects the computability of the loss function.  We will discuss in Section~\ref{sec:choose_Fk} how this choice can be made in any finite-dimensional local Hilbert space, including fermionic systems.\\
\\
\textbf{Estimability of $\mathcal{L}(\theta)$:} Even after the directions are fixed, the loss still depends on the unknown $H'$. This is remedied by an integration-by-parts identity that follows from the cyclicity of the trace [Eq.~\eqref{eq:app_IBP} in Appendix~\ref{app:std_loss}]:
\begin{align}
    \langle O\,\mathscr D_kH'\rangle_\rho   = -\langle [F_k,O]\rangle_\rho\;.
    \label{eq:IBP_main}
\end{align}

This identity replaces terms involving the
unknown target score $\mathscr D_kH'$ with expectation values of known
commutators. Consequently, up to a $\theta$-independent constant,
\begin{align}
    \mathcal L_{\rm std}(\theta)
    =
    \sum_k
    \left\langle
    |\mathscr D_kH_\theta|^2
    \right\rangle_\rho
    -
    2\Re
    \sum_k
    \langle
    [F_k^\dagger,\mathscr D_kH_\theta]
    \rangle_\rho .
    \label{eq:D_kloss}
\end{align}
See Appendix~\ref{app:std_loss} for a detailed derivation.\\
\\
\textbf{Approximation and Truncation of \texorpdfstring{$\mathscr{D}_kH_\theta$}{}:}\label{sec:approx}
A major practical hurdle in evaluating the loss function is computing the nonlinear quantity $\mathscr{D}_kH_\theta=[F_k,\e{-H_\theta}]\e{H_\theta}$ and the associated expectations.
Using the BCH identity, we have
\begin{align}
    \e{-H}F\,\e{H} =
    \sum_{\ell=0}^{\infty}\frac{(-1)^\ell}{\ell!}
    [H,F]_\ell\;,
\end{align}
and we can obtain the nested-commutator expansion for $\mathscr{D}_kH_\theta$ as
\begin{align}
    \mathscr{D}_kH_\theta
    &= F_k-\e{-H_\theta}F_k\e{H_\theta}
    = \sum_{\ell=1}^{\infty}\frac{(-1)^{\ell+1}}{\ell!}\,[H_\theta,F_k]_\ell\;,
    \label{eq:Dk_series}
\end{align}
where $[H,F]_\ell := [H,[H,\cdots[H,F]\cdots]]$ denotes the $\ell^{\rm th}$ nested commutator of $H$ with $F$.
Truncating \eqref{eq:Dk_series} at depth $d$ yields a polynomial in $\theta$ whose degree grows with $d$ and the commutator structure of the chosen operator family.  

In principle, one may replace the exponential by a polynomial approximation with improved uniform control. 
Namely, if $q_d(x)$ is a degree-$d$ polynomial approximating $\e{-x}$ on the relevant spectral interval, then one may define
\begin{align}
    \mathscr{D}^{(d)}_{k,\mathrm{poly}} H_\theta
    :=
    F_k - q_d(\operatorname{ad}_{H_\theta})F_k\;,
\end{align}
where $\operatorname{ad}_{H_\theta}(F):=[H_\theta,F]$. 
The flat-exponential approximation of Ref.~\cite{bakshi2023learningquantumhamiltonianstemperature} provides stronger approximation guarantees than the naive Taylor truncation. However, our numerical experiments show that the naive Taylor truncation suffices.

For practical purposes, we use a simple Taylor series truncation of degree $d$:
\begin{align}
    \mathscr{D}^{(d)}_kH_\theta
    := \sum_{\ell=1}^{d}\frac{(-1)^{\ell+1}}{\ell!}\,[H_\theta,F_k]_\ell\;.
\end{align}
This is the simplest approximation; in our experiments with exact density matrices it achieves rapid error decay until numerical precision saturates around degree $d \sim 20$. Empirically, \emph{we find that the depth required to reach a relative precision of $\epsilon$ depends only on the local degree and interaction strength, and not on the system size}. For example, numerical evidence is presented in Fig.~\ref{fig:tfim_dstar_vs_L}, where we plot the depth $d^*$ versus chain length for the TFIM and find that it depends on the inverse temperature $\beta$ or the coupling strength, and is independent of $L$ once the chain contains the full support of the light cone of nested commutators. This also holds for the other models studied.

Overall, our scheme reduces learning to evaluating expectations of a finite set of operator monomials under $\rho \propto \e{-H'}$. The dominant cost is the number of distinct commutator monomials generated by $[H_\theta,F_k]_\ell$, which depends on locality and the algebra of the chosen operator basis.

\begin{figure*}[t]
    \centering
    \includegraphics[width = \textwidth]{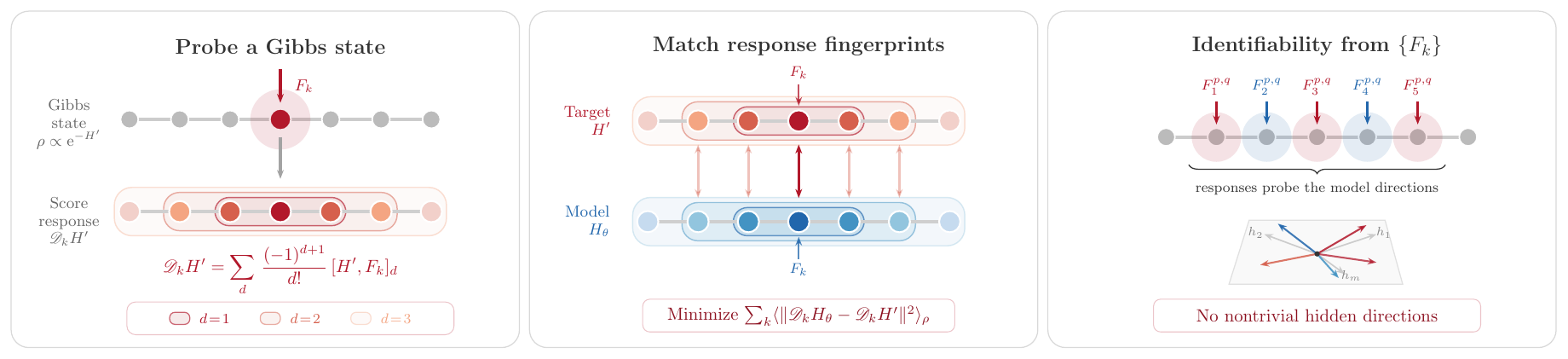}
    \caption{
    Schematic of quantum score matching through local probes $F_k$.
    \textbf{Left:} A local kick generated by an operator $F_k$ perturbs the Gibbs state $\rho$, producing a score response $\mathscr D_k H'$ that spreads around the lattice site with truncation depth $d$.
    \textbf{Middle:} Learning is performed by comparing the response fingerprints of the target Hamiltonian $H'$ and the model Hamiltonian $H_\theta$ under the same probes $F_k$ and minimizing the discrepancy.
    \textbf{Right:} Identifiability is controlled by probe coverage: when the set $\{F_k\}$ probes all nontrivial directions in the model class $\mathrm{span}\{h_i\}$, no hidden Hamiltonian direction remains invisible, and matching all responses determines $H_\theta$ up to an additive constant.}
\label{fig:learning_in_detail}
\end{figure*}

\subsection{Choosing Directions \texorpdfstring{$\{F_k\}$}{}} 
\label{sec:choose_Fk}
The choice of directions $\{F_k\}$ controls two competing
objectives:
 
\begin{enumerate}[label=(\roman*), wide, itemsep=2pt, labelindent=0pt]
  \item \textit{Completeness.} 
If the direction operators $\{F_k\}$ satisfy
\begin{align}
    \mathrm{Comm}(\{F_k\}) \cap \text{Alg}\{h_i\} \subseteq \mathrm{span}\{\mathds{1}\},
    \label{eq:ident_cond}
\end{align}
where $\text{Alg}\{h_i\}$ is the associative algebra generated by the operators $\{h_i\}$ and $\mathrm{Comm}(\{F_k\})$ is the set of operators commuting with $F_k$ for all $k$, then $\mathscr{D}_k H_\theta = \mathscr{D}_k H'$ for all $k$ implies $H_\theta - H' = c \cdot \mathds{1}$ for some constant $c$. This completeness condition states that multiples of the identity must be the only operators in the algebra generated by the Hamiltonian terms that also commute with all the direction operators. The full proof is given in Theorem~\ref{thm:identifiability} of Appendix~\ref{sec:identif_F}. This condition can obviously be satisfied by choosing the $F_k$ operators to generate the entire matrix algebra.
  
\item \textit{Computability.} Computation of the loss involves nested commutators of the form $[H_\theta, F_k]_\ell$. Hence, for the given operator basis ${h_k}$, the $F_k$ operators must be chosen so that these nested commutators are actually computable. For local Hamiltonians, these nested commutators spread over a region of radius $\sim\ell$ around $\mathrm{supp}(F_k)$.  Hence, it would be preferable to choose local $F_k$ operators in that case.
\end{enumerate}
 
\noindent We show that for any $n$-body system, a polynomial number of single-site direction operators
suffices to satisfy both conditions simultaneously.
We give the recipe in two cases set by particle statistics: distinguishable sites (spins, qubits, and bosons), where operators on different sites commute and each site is a tensor factor, and fermionic modes, where they anticommute. This completes the prescription for our approach, schematically depicted in Fig.~\ref{fig:learning_in_detail}.\\ 
\\
\textbf{Qudit and spin systems:} For an $L$-qubit system, we can choose single-qubit $Z \equiv \sigma^z$ and $X \equiv \sigma^x$ operators,
$F_k \in \{Z_1,X_1,\ldots,Z_L,X_L\}$, which generate the full $4^L$-dimensional
matrix algebra.



For a general $L$-site lattice with local Hilbert-space dimension $n$, the relevant object is the full single-site matrix algebra $M_n(\mathbb{C})$. For finite $n$, we therefore
use the generalized Pauli (clock--shift) operators $\hat C,\hat S$~%
\cite{Santhanam1976,Jagannathan1981,singh2020}, which satisfy
\begin{align}
    \hat S\hat C = \e{-2\pi\im/n}\,\hat C\hat S\;, \qquad
    \hat S^n = \hat C^n = \mathds{1}\;.
\end{align}
These operators generate $M_n(\mathbb{C})$ and reduce to Pauli matrices $Z,X$ for a local Hilbert space of dimension $n=2$, as shown in Appendix~\ref{sec:clock_shift}.

Another choice for these operators can be made by writing $\hat C = \e{\im a\,\hat p}$, $\hat S = \e{-\im b\,\hat q}$ with real $a,b$
satisfying $ab = 2\pi/n$, which yields finite-dimensional surrogates of the position and momentum operators $\hat p,\hat q$
that satisfy the canonical commutation relation in the limit of infinite local dimension,
\begin{align}
    \lim_{n\to\infty}[\hat p_{n,i},\hat q_{n,j}] = \im\,\delta_{ij}\;.
\end{align}
The direction operators are then chosen as
\begin{align}
    F_k \in \{\hat C_1,\hat S_1,\ldots,\hat C_L,\hat S_L\}\;,
    \label{eq:F_k_set}
\end{align}
(or the
$\{\hat p_i,\hat q_i\}$ surrogates), giving $2L$ directions with
$\mathrm{Comm}\{F_k\}=\mathbb{C}\mathds{1}$. Defining the commutant of a matrix $A$ as
\begin{align}
    \mathrm{Comm}(A) = \{X \in \mathbb{M}_n(\mathbb{C}):[A,X] = 0\}\;,
\end{align}
we find that the operators $F_k$ in Eq.~\eqref{eq:F_k_set} satisfy $\mathrm{Comm}\{F_k\}_k = \mathds{1}$, thus trivially satisfying the relation in Eq.~\eqref{eq:ident_cond}. However, one can choose the set of directions $\{F_k\}$ more judiciously and still satisfy the condition. Consider the example of the classical 1D Ising model,
\begin{align}
    H &= \sum_{i}J_{i,i+1}\sigma^z_i\sigma^z_{i+1} + \sum_i B_i\sigma^z_i\;,
\end{align}
for which ${\rm Alg}\{h_i\}$ lies entirely within the $\sigma^z$ subspace. Thus, in this case, one can choose the $F_k$ operators to be either $\sigma^x_j$ or $\sigma^y_j$ at each site; thus we need \emph{one} operator per site rather than two. In this sense, the choice in Eq.~\eqref{eq:F_k_set} can be too stringent in some cases, and one can make choices that are even more economical, based on the operator basis $\{h_i\}$.\\
\\
\textbf{Fermionic systems:} The same single-site philosophy extends to fermionic systems. By the Jordan--Wigner correspondence, the $2L_s$-mode Canonical Anticommuting Algebra is the
full matrix algebra $M_{2^{2L_s}}(\mathbb{C})$, so one \emph{can} always map a
fermionic problem to qubits and reuse the clock--shift (Pauli) directions above;
indeed, doing so recovers the Hamiltonian in our numerical tests. However, one can remain in the native fermionic formulation of the Hamiltonian and instead use the single-mode
Majorana operators
\begin{align}
    \gamma_{k,1} &= c_k^\dagger + c_k\;, \qquad k = 1,\dots,2L_s\ \notag\\
    \gamma_{k,2} &= \im(c_k^\dagger - c_k)\;,
\end{align}
where $c,c^\dagger$ are the local annihilation and creation operators, respectively, and the Majorana operators are Hermitian and obey the Clifford relation
$\{\gamma_{k,\alpha},\gamma_{l,\beta}\} = 2\delta_{kl}\delta_{\alpha\beta}$, the
anticommuting analog of the Weyl relation satisfied by $\hat C,\hat S$. The
$4L_s$ Majoranas generate the full Clifford algebra
$\mathrm{Cl}_{4L_s}(\mathbb{C}) \cong M_{2^{2L_s}}(\mathbb{C})$, so
$\mathrm{Comm}\{\gamma_{k,\alpha}\} = \mathbb{C}\mathds{1}$ and the completeness
condition~\eqref{eq:ident_cond} holds with no further assumptions, as shown in Appendix~\ref{sec:fermi_majorana}. Each $\gamma_{k,\alpha}$ is supported on a single fermionic mode,
where support is understood in the algebra generated by the local
creation and annihilation operators. 

To establish bounded support
growth, we additionally assume that
$H_\theta=\sum_X h_X(\theta)$ has finite interaction range $R$ and
that every local term has even fermion parity. An even fermionic
operator commutes with every operator on disjoint support,
including odd operators such as the Majorana
directions~\cite{nachtergaele2018fermions}. Consequently, each
nested commutator receives contributions only from Hamiltonian
terms overlapping the current support. It follows that
$[H_\theta,\gamma_{k,\alpha}]_\ell$ is supported within distance
$\ell R$ of mode $k$, as in the spin case.

These assumptions hold for the Fermi--Hubbard model considered
below, whose hopping, density, and interaction terms are all
even. Further details are given in
Appendix~\ref{sec:fermi_majorana}.

\subsection{Local PL inequality}\label{sec:pl_main}
The loss $\mathcal{L}_{\rm std}(\theta)=\sum_k \mathbb{E}|\mathscr{D}_kH_\theta-\mathscr{D}_kH'|^2$, despite having a single global minimum, is generally \emph{nonconvex} in $\theta$. In fact, by carefully tuning model parameters and the set of directions $\{F_k\}$, it is possible to construct explicit problem instances where the loss develops a local minimum. We provide the simplest example of this kind in Appendix~\ref{app:nonconvexity_example}. At the same time, with typical or random parameter initializations, we found that the loss often has a single minimum and is numerically easy to optimize.    

Near the global minimum, we show that a weaker condition, called the Polyak--\L{}ojasiewicz (PL) inequality, holds for the loss function in a neighborhood of $\theta^*$: 
\begin{align}
   \frac12 \| \nabla\mathcal{L}(\theta)\|^2 \ge   \mu\,(\mathcal{L}(\theta)-\mathcal{L}(\theta^*))\;,
\end{align}
for some $\mu>0$.
A detailed proof of the local PL inequality near $\theta^*$, including the Gram matrix lower bound and the local range condition, is given in Appendix~\ref{app:pl}.
For gradient descent with step size $\eta \leq 1/L$, where $L$ is the Lipschitz constant of $\nabla \mathcal{L}$, the PL inequality implies linear convergence \cite{polyak1963}:
\begin{align}
\mathcal{L}(\theta_k) - \mathcal{L}(\theta^*) \leq (1 - \eta \mu)^k \left(\mathcal{L}(\theta_0) - \mathcal{L}(\theta^*)\right).
\end{align}
The number of iterations to reach $\epsilon$-accuracy is $O(\mu^{-1} \log(1/\epsilon))$ \cite{karimi2016}. 
While our proof applies only in a local region around the minimum $\theta^*$, in numerical experiments with typical model instances, we find that the PL inequality holds for large radii around the minimum.


\subsection{Classical limits of the new loss}

\noindent \textbf{Application to discrete classical models:} For spin/qubit systems, with the direction operators chosen to be single-site Pauli operators, the loss function can also be written as 
\begin{align}
    \mathcal{L}_{\rm std}(\theta) &= 
    \sum_k \langle|F_k\e{-H_\theta}F_k \e{H_\theta} - F_k\e{-H'}F_k \e{H'}|^2\rangle_\rho \notag \\
    &\xrightarrow[]{\rm classical} \sum_k \left\langle\left|\frac{\e{-F_k H_\theta F_k}}{\e{-H_\theta}}-\frac{\e{-F_k H' F_k}}{\e{-H'}}\right|^2\right\rangle
\end{align}
which has the interpretation of matching the ratios of spin-flipped Hamiltonians due to the action of $F_k$ with the original Hamiltonian, in line with the philosophy of ratio-matching~\cite{hyvarinen2007}. In Appendix~\ref{app:std_loss}, Eq.~\eqref{eq:ratio_loss_app}, we show that this loss can be rewritten in the convenient form
\begin{align}
    \mathcal L_{\rm std}(\theta)
    &=
    \sum_k
    \left\langle
    \e{H_\theta}\e{-2F_kH_\theta F_k}\e{H_\theta}
    \right\rangle_\rho\notag\\
    &\qquad\qquad\quad-2\Re
    \sum_k
    \left\langle
    \e{F_kH_\theta F_k}\e{-H_\theta}
    \right\rangle_\rho .
    \label{eq:ratio_loss}
\end{align}
In particular, for the classical Ising model with nearest-neighbor couplings 
\begin{align}
    H_\theta &= -\sum_{u,v} J_{u,v}\sigma^z_u\sigma^z_{v} -\sum_u B_u\sigma^z_u\;,
\end{align}all operators
are diagonal in the $\sigma^z$ basis and commute, so the exponentials can be factored separately. Moreover, since the algebra generated by the Hamiltonian basis operators is diagonal, we can choose the direction operators to be just single-site $\sigma^x$ operators. Defining the local energy change under the spin flip at site $m$ ($F_m = \sigma^x_m$)
\begin{align}
    \Delta_m H_\theta
    &:=
    F_m H_\theta F_m - H_\theta\notag \\
    &=
    2\sigma^z_m\! \ (B_m + \sum_{v \neq m} J_{m,v} \sigma^z_v),
\end{align}
the loss \eqref{eq:ratio_loss} collapses to the compact local form
\begin{align}
    \mathcal{L}(\theta)
    &=
    \sum_m
    \langle
    \e{-2\Delta_m H_\theta}\rangle_\rho
    - 2\langle\e{\Delta_m H_\theta}
    \rangle_\rho,
    \label{eq:classical_Ising_limit}
\end{align}
where $\Delta_m H_\theta$ depends only on couplings involving sites connected to site $m$. If the underlying graph is known to be a lattice, the loss can be expressed in terms of local, few-body expectation values using a standard expansion of binary functions \cite{Shukla2025_LSFT,jayakumar2026computationally}. Each $\Delta_m H_\theta$ is linear in the parameters $\theta = (J,B)$, making the gradient straightforward to compute.

There are two interesting observations regarding \eqref{eq:classical_Ising_limit}. First, this loss function is new and different from previously known estimators for classical Ising models, including pseudolikelihood \cite{besag1975statistical} and interaction screening \cite{vuffray2016interaction, lokhov2018optimal}. Second, although the Ising model is discrete, this estimator is surprisingly derived using a method based on score matching, which normally only applies to distributions with continuous variables, by lifting the problem to the space of quantum models.\\
\\
\textbf{Application to continuous classical models:} The above loss function works for classical Hamiltonians as well. For example, in the classical continuous case, where $H(\mathbf{x}) = \sum_i \theta_i h_i(\mathbf{x})$ acts on configurations $\mathbf{x} = (x_1,x_2\cdots x_n)$, the quantum score reduces to the classical score in score matching. In particular, the action $F_k \circ \cdot = [F_k, \cdot]$ may be understood as the derivative operator $\partial_{{x}_k}\circ\cdot$. This can be made manifest in the case of the quantum $\hat{\phi}^4$ theory by taking $F_k = \hat{\pi}_k$, the conjugate momentum operator.  We have
\begin{align}
    [\hat{\pi}_k,\hat{\phi}_i^n] = \sum_{s=0}^{n-1}\hat{\phi}_i^{s}[\hat{\pi}_k,\hat{\phi}_i]\hat{\phi}_i^{n-s-1} = -\im n\hat{\phi}^{n-1}\delta_{ik}
\end{align}
which is exactly analogous to the simple classical identity $\partial_\phi \phi^n = n\phi^{n-1}$. In particular, we have
\begin{align}
    &\mathscr{D}_k (-\log(p_\theta)) = [F_k,\e{\log(p_\theta)}]\e{-\log(p_\theta)} = [F_k,\e{-H}]\e{H}\notag \\
    &[F_k,\e{-H}]\e{H} \mapsto \frac{\partial\e{-H}}{\partial{x}_k}\e{H} = -\frac{\partial H(\mathbf{x})}{\partial{x}_k} = \partial_k\log p(\mathbf{x})\;,
\end{align}
which makes explicit that $\mathscr{D}_k(-\log(p_\theta)) = \mathscr{D}_k H_\theta$ [Eq.~\eqref{eq:quantscore}] is the quantum analog of the \emph{classical} score $\nabla_k \log(p_\theta)$. Finally, the loss in the classical limit is
\begin{align}
    \mathcal{L}_{\rm std}(\theta)& = \sum_k\Tr(\rho |\mathscr{D}_kH_\theta - \mathscr{D}_kH'|^2) \notag \\
    &\xmapsto{\rm classical}\sum_k\int p(x) |\partial_k H_\theta - \partial_kH'|^2 dx\;,
\end{align}
and thus reduces exactly to the classical score matching objective, which is convex for classical models due to commutativity.

\section{RESULTS}
\subsection{Single-Qubit Case}
\begin{figure}[t]
    \centering
    \includegraphics[width=\linewidth]{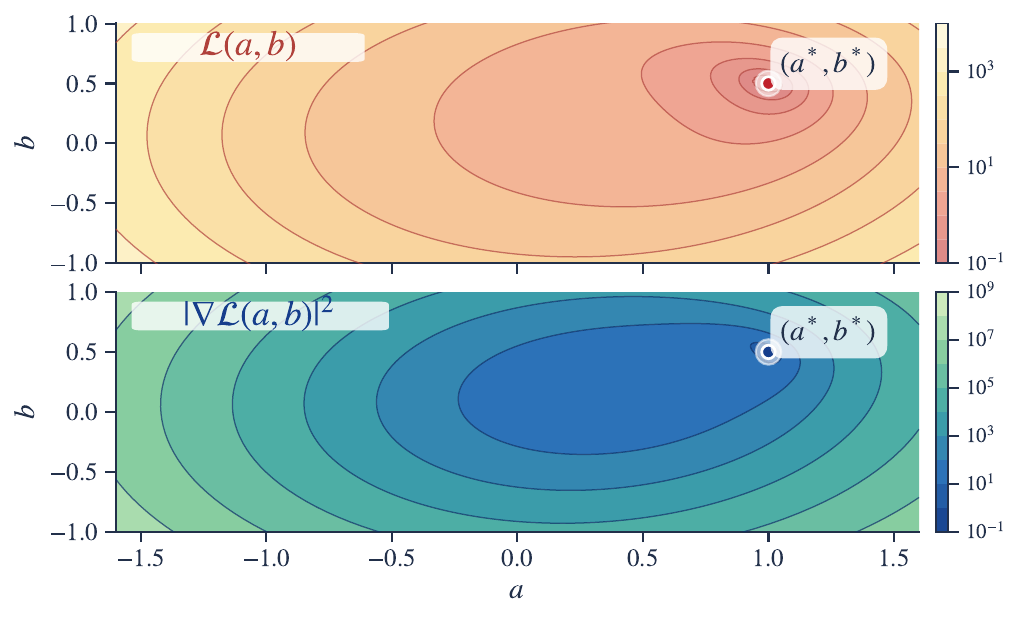}
    \caption{Loss landscape for a single qubit with $H = a\,\sigma^z + b\,\sigma^x$. \textbf{Top:} Contour plot of the loss $\mathcal{L}(a,b)$. \textbf{Bottom:} Contour plot of the norm squared of the gradient, $|\nabla \mathcal{L}(a,b)|^2$. As is evident, this loss has a unique minimum at the true couplings $(a^*,b^*) = (1.0,0.5)$.}
    \label{fig:single_qubit}
\end{figure}
It is instructive to demonstrate our method analytically in the simplest possible case of a single qubit with the 1d Hamiltonian
\begin{align}
    H_\theta = a\sigma^z+b\sigma^x\;.
\end{align}
In this case, choosing $F = \{\sigma^x,\sigma^z\}$ will give
the loss (using Eq.~\eqref{eq:ratio_loss})
\begin{align}
    \mathcal{L}(a,b) &= \mathbb{E}_\rho(\e{4H_\theta})-2\Re(\mathbb{E}_\rho(\e{-2H_\theta}))\notag \\
    &= (\cosh{4r}-2\cosh{2r})\notag \\
    &\qquad\qquad-\tau(\sinh{4r}+2\sinh{2r})
\end{align}
where we define 
\begin{align}
    r &= \sqrt{a^2+b^2}\notag \\
    \tau &= \tanh{r^*}\frac{aa^*+bb^*}{rr^*}\;,
\end{align}
where the asterisk $*$ denotes the true parameters. Working in polar coordinates $a = r\cos{\varphi}, b = r\sin{\varphi}$ and taking derivatives of the loss, we can see that
\begin{align}
    \partial_r\mathcal{L}(a,b) &= 8\cosh{3r}(\sinh{r}-\tau\cosh{r})\notag \\
    \partial_\varphi \mathcal{L}(a,b) &=  8\tanh{r^*}\sin(\varphi-\varphi^*)\sinh{r}\cosh^3{r}\;.
\end{align}
This loss has a unique minimum at $a=a^*$, $b=b^*$ (Fig.~\ref{fig:single_qubit}). It is seen here that the loss has only one minimum and that gradient descent has no trouble reaching this minimum.

\subsection{1D TFIM}
\begin{figure*}[t]
    \centering
    \begin{subfigure}[t]{0.48\textwidth}
        \centering
        \includegraphics[width=\linewidth]{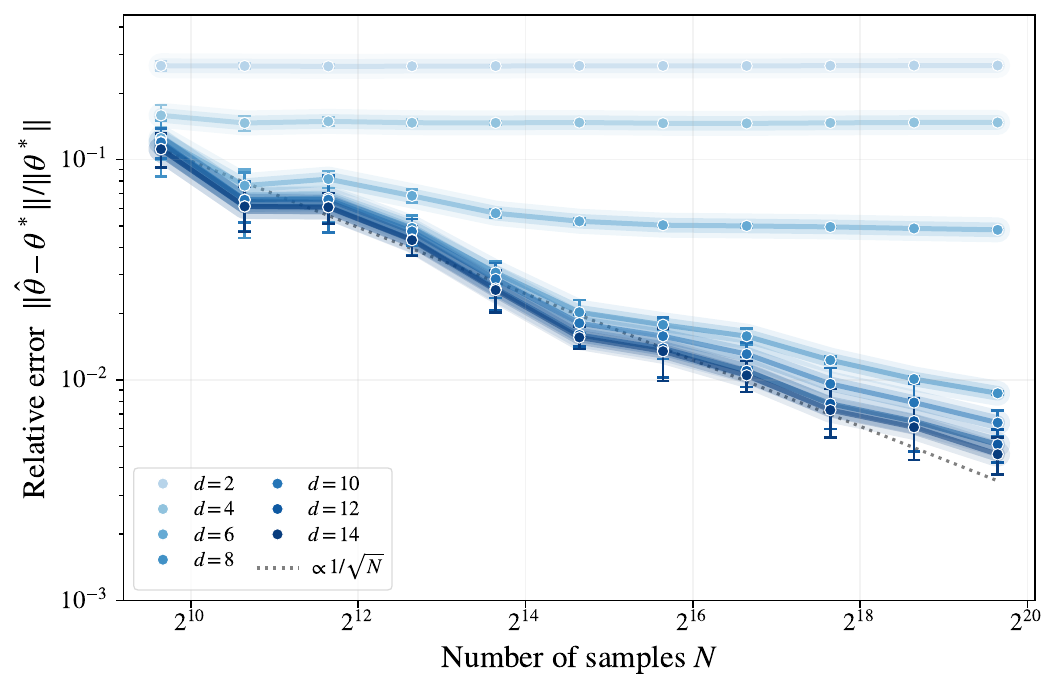}
        \caption{Relative errors for the $50$-qubit 1D non-uniform TFIM versus
        METTS sample size for various truncation depths $d$. As the depth
        increases, the sample size at which the error saturates increases.}
        \label{fig:nonuniform_tfim_trunc_err}
    \end{subfigure}\hfill%
    \begin{subfigure}[t]{0.48\textwidth}
        \centering
        \includegraphics[width=\linewidth]{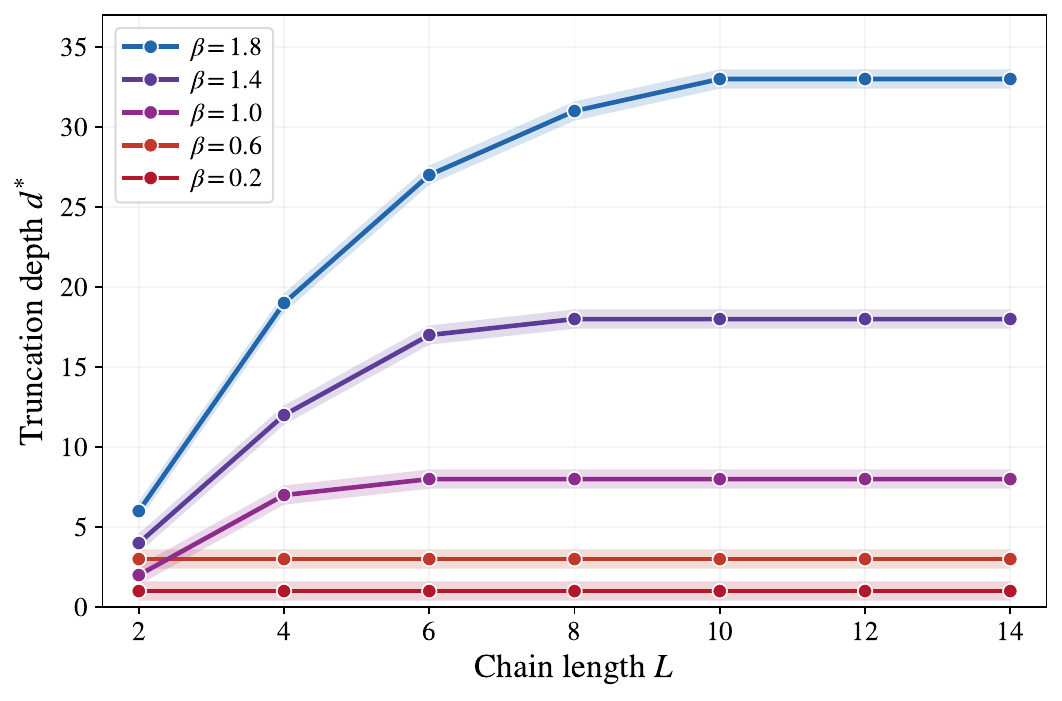}
        \caption{Minimal truncation depth $d^{*}$ for $10\%$ parameter recovery
        of the TFIM versus chain length $L$, for several inverse temperatures
        $\beta$. For each $\beta$, the asymptotic behavior of $d^{*}$ does not depend on $L$.}
        \label{fig:tfim_dstar_vs_L}
    \end{subfigure}%
    \caption{Truncation behavior of operator score matching for the 1D TFIM.}
    \label{fig:tfim_two_panel}
\end{figure*}
As a first multi-qubit example, we illustrate the method on the non-uniform 1D TFIM:
\begin{align}
    H = -\sum_i J_{i,i+1}\sigma^z_i\sigma^z_{i+1} -\sum_i B_i\sigma^z_i -\sum_i\Gamma_i\sigma^x_i\;.
\end{align}
Here, $\sigma^{z,x}_i$ are the Pauli matrices at site $i$ and the couplings are site-dependent: $J_{i,i+1}$ is the nearest-neighbor coupling, $h_i$ the longitudinal field, and $\Gamma_i$ the transverse field.
Choosing the derivative operator $F_k = \sigma^x_m,\sigma^z_m$ in the $x$ and $z$ directions at site $m$ yields
\begin{align}
    \mathscr{D}_{x,z}^{(d|m)}H &= [\sigma^{x,z}_m,\e{-H}]\e{H}\notag \\
    &\approx \sum_{k=1}^d \frac{(-1)^{k+1}}{k!}[H,\sigma^{x,z}_m]_k\;,
\end{align}
where, in the second line, we replace the full series by a degree-$d$ Taylor truncation. Each nested commutator $[H, \sigma^{x,z}_m]_\ell$ generates operators supported on a region of radius $\sim \ell$ around site $m$, reflecting the light-cone structure of information propagation. Since the $\sigma^{x},\sigma^z$ operators at each site are sufficient to span the complete matrix algebra, we see that all parameters of the model are effectively recovered by minimizing the loss function.

\paragraph*{Implementation:} For a chain of length $L = 50$, we work with a non-uniform TFIM model, with site-dependent couplings. Substituting into Eq.~\eqref{eq:D_kloss} and summing over all sites $m$ gives a degree-$2d$ loss function. We estimate the thermal expectation values entering the loss
in Eq.~\eqref{eq:D_kloss} using minimally entangled typical thermal
states (METTS)~\cite{white2009metts,stoudenmire2010metts} and minimize the resulting loss to recover the Hamiltonian parameters.

In Fig.~\ref{fig:nonuniform_tfim_trunc_err}, we show the total error in determining the site-dependent parameters $J_{ij},B_{i},\Gamma_i$ for various sample sizes and truncation depths of the BCH expansion. At small truncation depths, the error remains nearly flat even as the sample size increases, indicating that the dominant source of error is not statistical but systematic: the truncated loss is still a poor approximation to the true loss. As the depth is increased, this truncation error is reduced, and the behavior crosses over to the expected statistical regime, where the error decreases approximately as $1/\sqrt{N}$ with the number of samples $N$ before flattening out. This flattening is now due to noise: for large depths $d$, the statistical error from samples overwhelms any advantage due to increasing depth.

We thus find that for sufficiently large sample sizes and a sufficiently large truncation degree $d$, (i) the loss function recovers the correct parameters and (ii) small truncation degrees are sufficient to recover parameters in the physical parameter ranges studied. 

In Fig.~\ref{fig:tfim_dstar_vs_L}, we plot, for a fixed relative parameter
error $\epsilon = 0.1$, the truncation degree $d^*$ needed to recover
$\theta$ as a function of chain length $L$ for several inverse temperatures
$\beta$. The apparent growth of $d^*$ with $L$ at small sizes is a
finite-size artifact: the score operator is built from nested commutators
$\mathrm{ad}_H^{l} F$, whose support grows with the degree $l$, so on short
chains the support of these commutators reaches the boundary, and enlarging $L$ still enlarges
their effective support. Once $L$ exceeds the operator light cone at $\beta$, the commutators fit entirely within the
chain, adding sites no longer changes the error, and $d^*$ saturates.
Thus, $d^*$ is independent of $L$ and is fixed solely by the
strength and locality of the interaction (which we demonstrate by increasing $\beta$), confirming that the required truncation depth is a property of the local model rather than the total size of the system. We find that this qualitative behavior remains true for all models that follow. 

\begin{figure*}[!t]
    \centering
    \begin{minipage}[t]{0.48\textwidth}
        \centering
        \includegraphics[width=\linewidth]{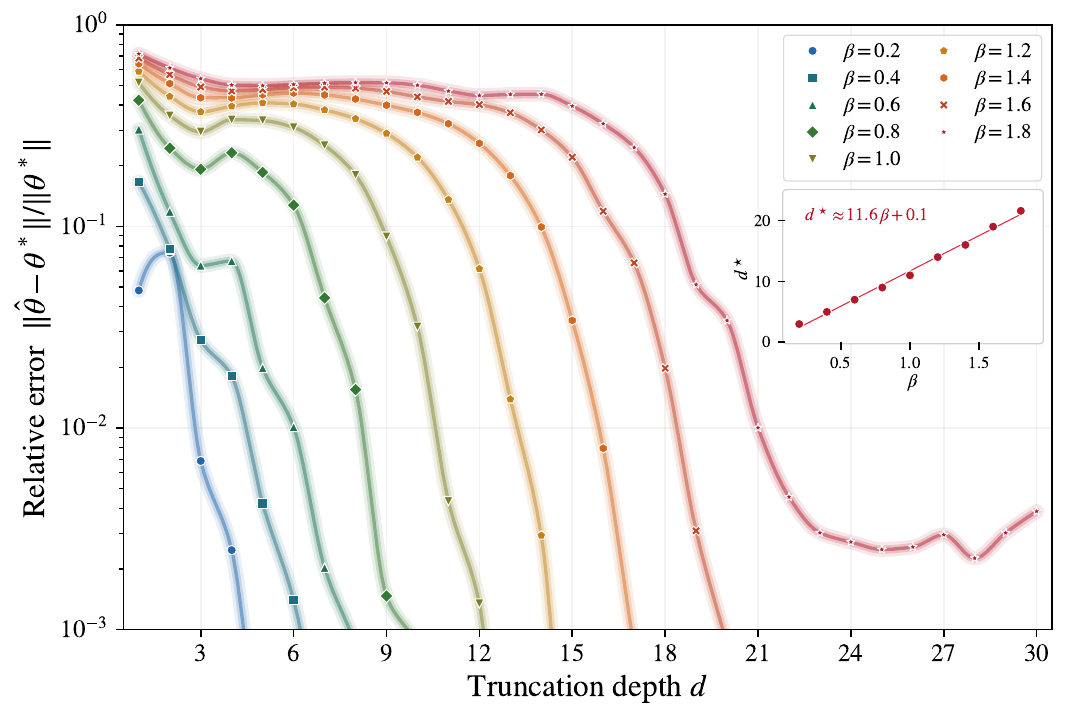}
        \caption{Mean relative errors versus depth $d$ for various values of
        $\beta$ for the XXZ model with $L=4$ sites, with an inset plot of the depth
        $d^*$ required to reach a relative error of $1\%$. We find that the
        required depth $d_{1\%} \propto \beta$.}
        \label{fig:xxz_beta_depth_sweep}
    \end{minipage}
    \hfill%
    \begin{minipage}[t]{0.48\textwidth}
        \centering
        \includegraphics[width=\linewidth]{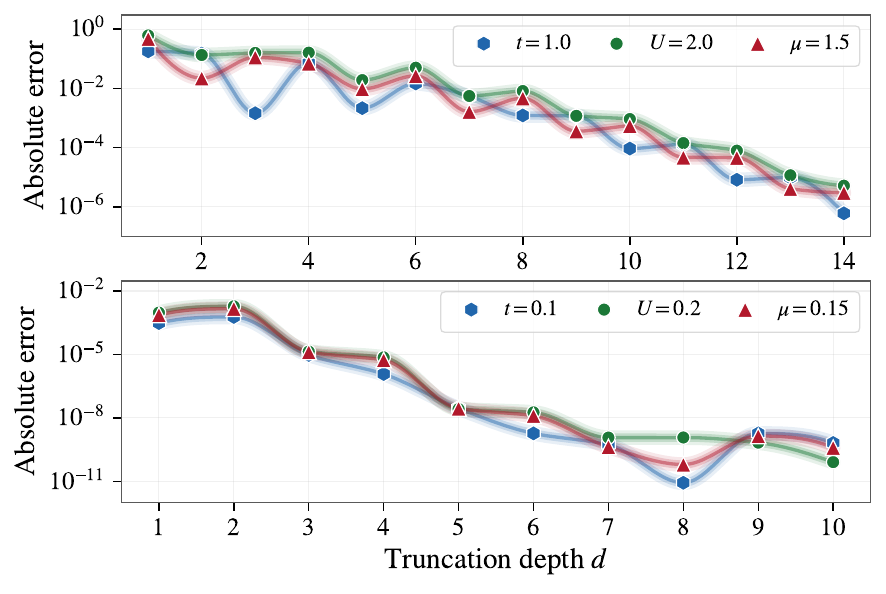}
        \caption{Absolute errors in the parameters of the 1D Fermi--Hubbard model with $L=2$ sites
        as a function of truncation degree $d$ for two choices of exact
        parameters differing by a factor of $10$.
        \begin{tabular}{@{}l@{\ }l@{}}
        \textbf{Top:}    & $(t,U,\mu)=(1.0,\,2.0,\,1.5)$ \\
        \textbf{Bottom:} & $(t,U,\mu)=(0.1,\,0.2,\,0.15)$.
        \end{tabular}}
        \label{fig:fh_trunc_err}
    \end{minipage}%
\end{figure*}
\subsection{1D XXZ Model}
We demonstrate the method on the 1D XXZ spin chain:
\begin{align}
    H &= -J_{xy}\sum_i (\sigma^x_i\sigma^x_{i+1} + \sigma^y_i\sigma^y_{i+1}) \notag \\
        &\qquad-J_z\sum_i \sigma^z_i\sigma^z_{i+1} 
        -B\sum_i \sigma^z_i\;.
\end{align}
The XXZ model interpolates between the isotropic Heisenberg model ($J_{xy} = J_z$) and the Ising model ($J_{xy} = 0$), making it an ideal testbed for validating parameter recovery across different physical regimes. Here, $J_{xy},J_z$ are nearest-neighbor couplings, and $B$ is the magnetic field. We choose single-site Pauli operators as the directions $F$. \\
\paragraph*{Implementation:} To generate synthetic training data with known true parameters $\theta^* = (J_{xy}^*, J_z^*, B^*)$, we prepare the finite-temperature Gibbs state $\rho(\theta^*) \propto \e{-H(\theta^*)}$ for $L=4$ sites. For fixed parameters $(J_{xy},J_z,B) = (0.8,0.3,0.2)$, we vary $\beta$ from $0.2$ to $1.8$ and measure the relative errors at each depth. We also record the minimum depth $d$ for each $\beta$ that results in a relative error of $1\%$, as shown in Fig.~\ref{fig:xxz_beta_depth_sweep}. We find that the required depth $d_{1\%} \propto \beta$, confirming that larger values of $\beta$ require greater depth to reach the required precision.

\subsection{1D Fermi--Hubbard Model}

To demonstrate the framework for fermionic systems, consider the Fermi--Hubbard model
\begin{align}
    H &= -t\sum_{\langle ij\rangle\sigma}(c_{i\sigma}^\dagger c_{j\sigma} + h.c.) + U\sum_i n_{i\uparrow}n_{i\downarrow} - \mu\sum_{i\sigma}n_{i\sigma}
\end{align}
with parameters $\theta = (t, U, \mu)$. Here, $c_{i\sigma}^\dagger$ and $c_{i\sigma}$ are fermionic creation and annihilation operators satisfying $\{c_{i\sigma}, c_{j\sigma'}^\dagger\} = \delta_{ij}\delta_{\sigma\sigma'}$, $t$ is the nearest-neighbor hopping amplitude, $U$ is the on-site Coulomb repulsion, $\mu$ is the chemical potential, and $n_{i\sigma} = c_{i\sigma}^\dagger c_{i\sigma}$ is the number operator. The Hamiltonian is linear in the parameters: $H_\theta = t\, h_{\text{hop}} + U\, h_{\text{int}} - \mu\, h_{\mu}$, where $h_{\text{hop}}$, $h_{\text{int}}$, and $h_\mu$ are fixed Hermitian operators.

For $L_s$ sites with two spin species, the Fock space has dimension $2^{2L_s}$, and all operators act as matrices on this space. For this model, we choose the direction operators $F_k$ to be the Majorana fermions
\begin{align}
    \gamma_{k,1} = c_k^\dagger + c_k\;, \qquad \gamma_{k,2} = i(c_k^\dagger - c_k)\;,
\end{align}
where $k$ ranges over all $2L_s$ fermionic modes (sites $\times$ spins). These are Hermitian and satisfy $\{\gamma_{k,\alpha}, \gamma_{l,\beta}\} = 2\delta_{kl}\delta_{\alpha\beta}$. The Majorana directions are the natural choice for fermionic systems: the commutator $[\gamma_{k,1}, H]$ probes how the Hamiltonian couples to single-particle excitations at mode $k$, analogous to how $[\sigma^y_m, H]$ probes spin fluctuations in the TFIM. For $2L_s$ modes, the $4L_s$ Majorana operators generate the full Clifford algebra $\mathrm{Cl}(4L_s, \mathbb{R}) \cong M_{2^{2L_s}}(\mathbb{C})$, so the identifiability condition is automatically satisfied: only multiples of the identity commute with all the operators $F_k$.

\paragraph*{Implementation.} For $L_s = 2$ (four fermionic modes), we have three parameters and eight Majorana direction operators. We construct the exact thermal state $\rho \propto \e{-H(\theta^*)}$ with true parameters $(t^*, U^*, \mu^*) = (1.0, 2.0, 1.5)$ and $(0.1, 0.2, 0.15)$ and minimize the truncated loss over~$\theta$. Figure~\ref{fig:fh_trunc_err} shows the parameter recovery errors as a function of truncation depth~$d$. All three parameters exhibit factorial error decay and reach errors of $\sim 10^{-6}$ even when the parameters are $O(1)$, despite the more complex operator algebra.


\subsection{1D Lattice \texorpdfstring{$\phi^4$}{} Theory}\label{sec:phi4}
To demonstrate applicability beyond spin systems, we consider lattice scalar field theory:
\begin{align}
    H = \sum_{i\in\Lambda}\left(\frac{1}{2}\pi_i^2 + \frac{m^2}{2}\phi_i^2
        + \frac{\lambda}{4!}\phi_i^4\right)
        + \kappa\sum_{\langle ij\rangle}(\phi_i-\phi_j)^2\;,
\end{align}
where $\phi_i$ and $\pi_i$ are conjugate field and momentum operators satisfying $[\phi_i, \pi_j] = i\delta_{ij}$. The parameters $\theta = (\kappa, m^2, \lambda)$ control the spatial correlation term, mass, and quartic self-interactions.

In the continuum limit, $\phi_i$ and $\pi_i$ act on infinite-dimensional Hilbert spaces and are not trace-class, invalidating the trace identities used in Sec.~\ref{sec:approx}. We therefore employ a Fock-space truncation: at each site, we represent the field operators in a harmonic oscillator basis
\begin{align}
    \phi_i = \frac{1}{\sqrt{2\omega}}(a_i + a_i^\dagger)\;, \qquad
    \pi_i = i\sqrt{\frac{\omega}{2}}(a_i^\dagger - a_i)\;,
\end{align}
with $a_i$ truncated to a $d_{\rm loc} \times d_{\rm loc}$ matrix (keeping at most $d_{\rm loc}-1$ excitations per site). The frequency $\omega$ is a variational parameter that can be optimized to minimize truncation effects.

\begin{figure}[!t]
        \includegraphics[width=\linewidth]{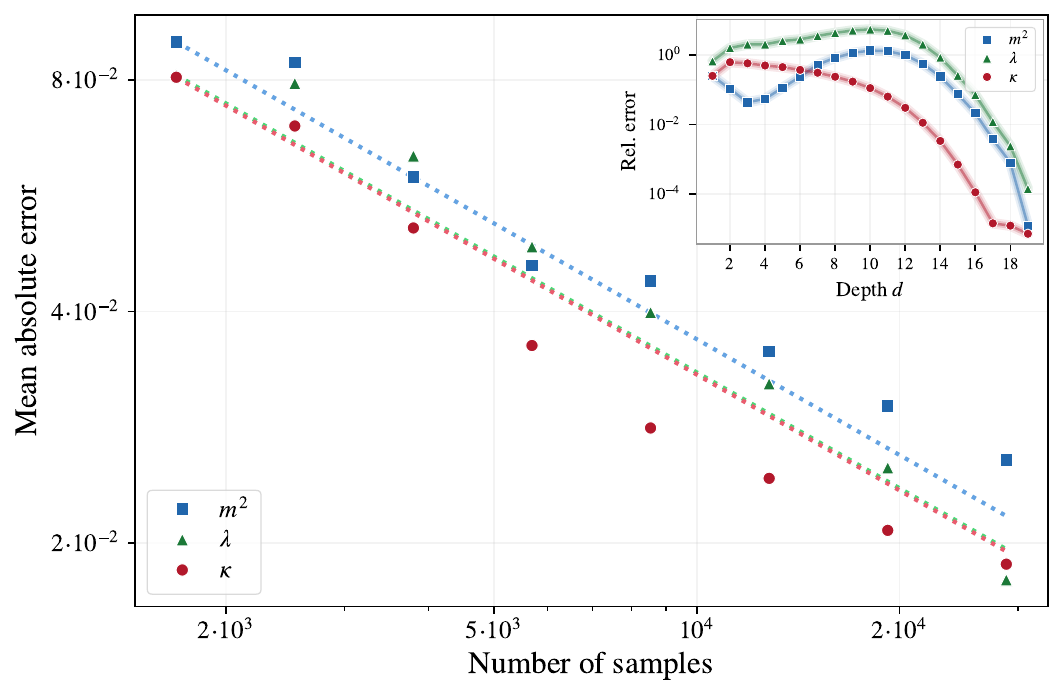}
        \caption{Lattice $\phi^4$ theory convergence: errors in the parameters $m,\kappa,\lambda$
        for the 1D $\phi^4$ model with $L=5$ sites versus the number of METTS samples. Inset: relative
        error in the parameters versus truncation depth $d$. In both cases, as the
        number of samples or the depth increases, the error decreases.}
        \label{fig:phi4_samples_err}
\end{figure}

We choose the directions $F_k = \phi_k,\pi_k$ (the field and momentum operators at each site). The score operator associated with $\phi_k$ is
\begin{align}
    \mathscr{D}_{\phi_k} H = [\phi_k, e^{-H}]e^{H}\;.
\end{align}
The commutator structure reflects the locality of the Hamiltonian: $[\phi_k, \pi_j^2] \propto \delta_{kj}$ and $[\phi_k, \phi_j^n] = 0$, so the score operator couples only to the momentum term at the same site and to gradient terms involving neighboring sites.
\paragraph*{Implementation.}
We work on $L=5$ sites with periodic boundaries and true parameters
$(m^{2*}, \lambda^*, \kappa^*,) = (0.4,0.5,0.3)$, representing the field
operators in a harmonic-oscillator basis. The harmonic-oscillator basis operators ($a,a^\dagger$) are infinite-dimensional, so we truncate them to matrices of dimension $n_{\rm loc}=8$. Using the directions
$F_k \in \{\phi_i,\pi_i\}$, we expand each score operator to depth $d$
[Eq.~\eqref{eq:Dk_series}] and assemble the loss [Eq.~\eqref{eq:D_kloss}]. The
main panel of Fig.~\ref{fig:phi4_samples_err} uses METTS~\cite{white2009metts,stoudenmire2010metts} samples of
$\rho \propto \e{-H(\theta^*)}$ for the full loss; the error follows the expected
$1/\sqrt{N_{\rm samp}}$ scaling. The
inset uses the exact $\rho$ and sweeps the truncation depth $d$, isolating the
systematic error.

We note that the local truncation dimension $n_{\rm loc}$ must be relatively large here: at small $n_{\rm loc}$, the
operators $\phi^2$ and $\phi^4$ are nearly collinear, making $m^2$ and $\lambda$ harder to probe. 

\section{DISCUSSION}\label{sec:discussion}
We have introduced a quantum (operator) score-matching framework for Hamiltonian learning
that avoids partition function computation by exploiting non-commutative
integration-by-parts identities. The global minimum of the novel loss recovers Hamiltonian parameters under mild algebraic conditions on the direction operators. We have provided specific prescriptions for choosing directions that satisfy these conditions and illustrated that the particular choice influences the loss landscape. Our approach complements existing methods based on Bayesian inference~\cite{wiebe2014}, local reduced density matrices~\cite{bairey2019}, and time evolution~\cite{anshu2021}, with the
distinction that it operates directly on thermal states without requiring
coherent dynamics.

The truncated score ${\mathscr{D}}^{(d)}_k H_\theta$ is a polynomial of
degree $d$ in $\theta$, making the loss a degree-$2d$ polynomial, which makes evaluation of the loss gradient tractable and scalable. We empirically observed that the depth $d$ required to reach a relative precision depends only on the local degree and interaction strength, and not on the system size. This polynomial structure connects naturally to the sum-of-squares hierarchy~\cite{barak_sos,FKP19} exploited by BLMT~\cite{bakshi2023learningquantumhamiltonianstemperature}, who solve a convex SoS relaxation with provable worst-case guarantees. Our approach instead uses gradient descent on the nonconvex loss, which is cheaper per iteration and numerically effective.
Across the numerical experiments reported in this paper, we observed that the loss often has a single minimum and is numerically easy to optimize.

We proved that the loss function has nice local convexity properties near the global minimum. A complete non-asymptotic error analysis of our method remains an open problem. Such an analysis would need to address the generally nonconvex loss and propagate both the BCH truncation error and the statistical error from finite samples to obtain end-to-end parameter-recovery and sample-complexity guarantees. We leave this analysis to future work.

The Fermi--Hubbard example confirms that the framework is agnostic to particle statistics: no Jordan--Wigner or other qubit mapping is required, and the same error scaling applies to spins, bosons, and fermions. Depending on the setting, the expectation values entering the loss can be estimated via tensor network methods, classical shadow tomography~\cite{Huang_2020}, or direct quantum measurements.

Several extensions and applications of our approach would be interesting to consider: continuum quantum field theories (where
trace identities require regularization), learning from dynamical
measurements, and applications to lattice gauge theories. We leave further analysis of these extended settings to future work.

\section*{DATA AND CODE AVAILABILITY}

All data and code used to generate the numerical results and figures presented in this paper are publicly available in the project’s \href{https://github.com/shreyaalkuhs/Learning-of-Quantum-Hamiltonians-via-Operator-Score-Matching}{GitHub repository} \cite{osm_code}.

\section*{ACKNOWLEDGMENTS}
We thank our colleagues Scott Lawrence, Yukari Yamauchi, Sidhant Misra, and Marc Vuffray at Los Alamos National Laboratory for helpful discussions. We acknowledge support from the U.S.~Department of Energy/Office of Science Advanced Scientific Computing Research Program. Los Alamos National Laboratory is operated by Triad National Security, LLC, for the National Nuclear Security Administration of U.S.~Department of Energy, under Contract No.~89233218CNA000001.

\bibliographystyle{plain}
\bibliography{ref_draft}

\onecolumngrid
\appendix

\medskip

\section{Standard-metric objective and trace integration by parts}
\label{app:std_loss}
\newlength{\oldparindent}%
\setlength{\oldparindent}{\parindent}
\setlength{\parindent}{0em}

This appendix derives the standard-metric score-matching loss used in the main text and shows how to rewrite it in a form that depends only on expectations under the data state, without explicit dependence on the partition function.

\medskip

\subsection{Setup and score map}
Let $\rho := \e{-H'}/Z'$ with $Z'=\Tr(\e{-H'})$. Define the (noncommutative) score map for a direction operator $F_k$ by
\begin{align}
    \mathscr{D}_k H \;:=\; [F_k,\e{-H}]\,\e{H}
    \;=\; F_k - \e{-H}F_k \e{H}.
    \label{eq:app_score_map}
\end{align}
Let $H_\theta=\sum_{i=1}^m \theta_i h_i$ be the model Hamiltonian and define the residual
\begin{align}
    \Delta_k(\theta) := \mathscr{D}_k H_\theta - \mathscr{D}_k H'.
\end{align}
We use the $\rho$-weighted Hilbert--Schmidt norm
\begin{align}
    \mathbb{E}_{\rho}(X) := \Tr(\rho\,X),
    \qquad
    \|A\|_{\mathrm{HS}}^2:=\Tr(A^\dagger A),
\end{align}
and define the standard-metric objective
\begin{align}
    \mathcal{L}_{\mathrm{std}}(\theta)
    := \sum_{k=1}^n \mathbb{E}_{\rho}\!\left(\Delta_k(\theta)^\dagger\Delta_k(\theta)\right).
    \label{eq:app_Lstd_def}
\end{align}
\subsubsection*{Special case}
Note that when $F_k$ is a Pauli matrix, $\norm{\Delta_k}^2 = \norm{F_k \Delta_k}^2$, and the loss can be written as 
\begin{align}
    \mathcal{L}_{\rm ratio}(\theta) &= \Tr(\rho \|\underbrace{F_k\e{-H_\theta}F_k\e{H_\theta}}_{= \mathcal{D}_kH_\theta}-F_k\e{-H'} F_k\e{H'}\|^2)
\end{align}
where we call this form of the loss function the ratio-matching loss, following its classical counterpart~\cite{hyvarinen2007}. We also define the new ``derivative'' as
\begin{align}
    \mathcal{D}_k H &= F_k\e{-H}F_k\e{H} = F_k(F_k-\mathscr{D}_k H) = \mathds{1}-F_k\mathscr{D}_kH\;.
\label{eq:newD_oldD}    
\end{align}

\medskip

\subsection{Trace integration-by-parts identity}

For trace-class operators $A,B$, one has $\Tr([A,B])=0$. In particular, for any trace-class $O$,
\begin{align}
    0 &= \Tr\!\left([F_k, O\,\e{-H'}]\right)
    = \Tr\!\left([F_k,O]\,\e{-H'}\right) + \Tr\!\left(O\,[F_k,\e{-H'}]\right).
\end{align}
Divide by $Z'$ to write this as a $\rho_{\mathrm{data}}$-expectation:
\begin{align}
    \Tr\!\left(O\,[F_k,\e{-H'}]\right)
    &= 
    - \Tr\!\left([F_k,O]\e{-H'}\right)\notag \\
    \implies \Tr\!\left(O\,[F_k,\e{-H'}]\e{H'}\frac{\e{-H'}}{Z'}\right)
    &= - \Tr\!\left([F_k,O]\frac{\e{-H'}}{Z'}\right)
    \label{eq:IBP}
\end{align}
Using $\mathscr{D}_k H' = [F_k,\e{-H'}]\e{H'}$ and $\e{-H'}=Z'\rho$ gives the useful form
\begin{align}
    \Tr\!\left(O\,\mathscr{D}_k H'\rho\right)
    &=  
    -\Tr\!\left([F_k,O]\rho\right)\notag \\
    \implies \mathbb{E}_\rho\!\left(O\,\mathscr{D}_k H'\right)
    &=
    -\mathbb{E}_\rho\!\left([F_k,O]\right)\;.
    \label{eq:app_IBP}
\end{align}

\medskip

\subsubsection{Special case}
One can also generalize this IBP identity to $\mathcal{D}_k H_\theta$. We can use Eq.~\eqref{eq:newD_oldD} to write for any trace-class operator $O$
\begin{align}
     \mathbb{E}_\rho(O\mathcal{D}_kH') &= \mathbb{E}_\rho(O) -\mathbb{E}_\rho(OF_k\mathscr{D}_kH')\;.
\end{align}
Now apply IBP from Eq.~\eqref{eq:app_IBP} to $O \rightarrow OF_k$
\begin{align}
    \mathbb{E}_\rho(O\mathcal{D}_kH') &= \mathbb{E}_\rho(O) + \mathbb{E}_\rho( [F_k,OF_k]) \notag \\
    \mathbb{E}_\rho(O\mathcal{D}_kH') &= \mathbb{E}_\rho(F_k OF_k)\;.
    \label{eq:specialIBP}
\end{align}

\medskip

\subsection{Loss rewrite in terms of model score only}
Expand \eqref{eq:app_Lstd_def}:
\begin{align}
    \mathcal{L}_{\mathrm{std}}(\theta)
    &= \sum_k \mathbb{E}_{\rho}\!\left(
    (\mathscr{D}_k H_\theta)^\dagger(\mathscr{D}_k H_\theta)\right)
    - 2\,\Re\,\sum_k \mathbb{E}_{\rho}\!\left(
    (\mathscr{D}_k H_\theta)^\dagger\,\mathscr{D}_k H'\right)
    + C,
    \label{eq:app_expand}
\end{align}
where $C:=\sum_k\mathbb{E}_{\rho}((\mathscr{D}_k H')^\dagger\mathscr{D}_kH')$ is independent of $\theta$. Apply \eqref{eq:app_IBP} with $O=(\mathscr{D}_k H_\theta)^\dagger$:
\begin{align}
    \mathbb{E}_{\rho}\!\left(
    (\mathscr{D}_k H_\theta)^\dagger\,\mathscr{D}_k H'\right)
    = - \mathbb{E}_{\rho}\!\left([F_k,(\mathscr{D}_k H_\theta)^\dagger]\right).
\end{align}
Using $\Re \langle X \rangle = \Re \langle X^\dagger \rangle$ and $[F_k,(\mathscr{D}_k H_\theta)^\dagger]^\dagger=-[F_k^\dagger,\mathscr{D}_k H_\theta]$ yields
\begin{align}
    \mathcal{L}_{\mathrm{std}}(\theta)
    &= \sum_k \mathbb{E}_{\rho}\!\left(
    (\mathscr{D}_k H_\theta)^\dagger(\mathscr{D}_k H_\theta)\right)
    - 2\,\Re\,\sum_k \mathbb{E}_{\rho}\!\left([F_k^\dagger,\mathscr{D}_k H_\theta]\right)\;,
    \label{eq:app_Lstd_rewrite}
\end{align}
where we have dropped the $\theta$-independent term. Equation \eqref{eq:app_Lstd_rewrite} is the main computational form: it depends on $\rho$ only through expectations and depends on the model only through $\mathscr{D}_kH_\theta$.

\medskip

\subsubsection{Special Case of Pauli Strings}
\label{app:ratio_loss_app}
Again, we expand the loss in the special case where the $F_k$ are Pauli strings, and we call this the ratio loss:
\begin{align}
    \mathcal{L}_{\rm ratio}(\theta) &= \mathbb{E}_\rho(\norm{\mathcal{D}_kH_\theta - \mathcal{D}_k H'}^2)\notag \\
    &= \mathbb{E}_\rho(\norm{\mathcal{D}_kH_\theta}^2) -2\Re(\mathbb{E}_\rho((\mathcal{D}_kH_\theta)^\dagger(\mathcal{D}_kH'))) + C'
\end{align}
where now $C' = \mathbb{E}_\rho(\norm{\mathcal{D}_kH'}^2)$ is a $\theta$-independent term. Using the special-case IBP identity in Eq.~\eqref{eq:specialIBP} for $O = (\mathcal{D}_kH_\theta)^\dagger$, we get 
\begin{align}
    \mathcal{L}_{\rm ratio}(\theta) &= \mathbb{E}_\rho(\norm{\mathcal{D}_kH_\theta}^2) - 2\Re(\mathbb{E}_\rho(F_k(\mathcal{D}_kH_\theta)^\dagger F_k))\notag \\
    &= \mathbb{E}_\rho(\e{H_\theta}F_k\e{-H_\theta}F_k^2\e{-H_\theta}F_k\e{H_\theta})-2\Re(\mathbb{E}_\rho(F_k\e{H_\theta}F_k\e{-H_\theta}F_k^2))\notag \\
    &= \mathbb{E}_\rho(\e{H_\theta}F_k\e{-2H_\theta}F_k\e{H_\theta})-2\Re(\mathbb{E}_\rho(F_k\e{H_\theta}F_k\e{-H_\theta}))\notag\\
    &= \mathbb{E}_\rho(\e{H_\theta}\e{-2F_kH_\theta F_k}\e{H_\theta})-2\Re(\mathbb{E}_\rho(\e{F_kH_\theta F_k}\e{-H_\theta}))\;.
    \label{eq:ratio_loss_app}
\end{align}

For a Hermitian unitary operator $F_k$ ($F_k^\dagger = F_k$, $F_k^2 = \mathds{1}$),
\eqref{eq:newD_oldD} gives $\mathcal{D}_k H = \mathds{1} - F_k\mathscr{D}_k H$,
so that
\begin{align}
    \mathcal{D}_kH_\theta - \mathcal{D}_kH' = -F_k\big(\mathscr{D}_kH_\theta - \mathscr{D}_kH'\big)\;.
\end{align}
Since $\|F_kX\|_{\rm HS}^2 = \Tr(X^\dagger F_k^\dagger F_k X) = \|X\|_{\rm HS}^2$,
the $\rho$-weighted Hilbert--Schmidt norms of the two residuals are identical
and therefore $\mathcal{L}_{\rm std}(\theta) = \mathcal{L}_{\rm ratio}(\theta)$.
The expressions \eqref{eq:app_Lstd_rewrite} and \eqref{eq:ratio_loss} are
two computational forms of the same loss, obtained by applying IBP to two
different forms of the same underlying conjugated direction
$\e{-H}F_k\e{H}$.

\medskip

\subsection{Nested-commutator series and truncation}
Using
\begin{align}
    \e{-H}F\,\e{H} = \sum_{\ell=0}^\infty \frac{(-1)^\ell}{\ell!}[H,F]_\ell,
\end{align}
we obtain
\begin{align}
    \mathscr{D}_k H
    = F_k - \e{-H}F_k \e{H}
    = \sum_{\ell=1}^{\infty}\frac{(-1)^{\ell+1}}{\ell!}\,[H,F_k]_\ell.
    \label{eq:app_D_series}
\end{align}
A depth-$d$ Taylor surrogate is
\begin{align}
    \mathscr{D}_k^{(d)} H
    := \sum_{\ell=1}^{d}\frac{(-1)^{\ell+1}}{\ell!}\,[H,F_k]_\ell.
    \label{eq:app_D_trunc}
\end{align}
Substituting $\mathscr{D}_k^{(d)}H_\theta$ into \eqref{eq:app_Lstd_rewrite} yields an approximation to the loss that is polynomial in $\theta$, with a degree that grows with $d$ and the operator algebra.

\medskip

\section{Choice of the directions $\{F_k\}$ in the loss}
\label{app:choice_of_directions}

\medskip

\subsection{Identifiability}
\label{sec:identif_F}

In this section, we study the minimal requirements on the set of directions 
$\{F_k\}$ needed to ensure identifiability. In Sec.~\ref{sec:gm_upd}, we adopted the strong assumption \eqref{assump:S3} that the directions generate the full matrix algebra. Although sufficient, this assumption is not necessary. We show that identifiability can be guaranteed by a smaller subset of directions $\{F_k\}$ indexed by a set $\mathcal{K} \ni k$, and we derive a sufficient condition for such minimal direction sets.
We begin with some definitions and identities.
\begin{definition}[Commutant]
\label{def:comm}
Let $\mathcal{A} \subseteq \mathcal{B}(\mathcal{H})$ be a set of operators on a Hilbert space $\mathcal{H}$. The \emph{commutant} of $\mathcal{A}$, denoted $\mathrm{Comm}(\mathcal{A})$, is
\begin{align}
    \mathrm{Comm}(\mathcal{A}) := \{X \in \mathcal{B}(\mathcal{H}) : [X, A] = 0 \;\; \forall A \in \mathcal{A}\}.
\end{align}
\end{definition}

\begin{definition}[Hamiltonian subspace and generated algebra]
\label{def:ham_subspace}
Given a parametrized Hamiltonian $H_\theta = \sum_{i=1}^m \theta_i h_i$ with
linearly independent operators $\{h_i\}_{i=1}^m$, define the Hamiltonian
subspace $\mathcal{S}$ and the unital associative algebra
$\mathrm{Alg}(\mathcal{S})$ it generates,
\begin{align}
    \mathcal{S} &:= \mathrm{span}_{\mathbb{R}}\{h_i\}_{i=1}^m \subset \mathcal{B}(\mathcal{H}),\notag\\
    \mathrm{Alg}(\{h_i\}) \equiv \mathrm{Alg}(\mathcal{S}) &:= \mathrm{span}_{\mathbb{C}}
    \big\{\, h_{i_1} h_{i_2}\cdots h_{i_p} \;:\; p \ge 0,\ i_j \in \{1,\dots,m\} \,\big\},
\end{align}
where the products are taken under ordinary operator multiplication and the
empty product ($p=0$) is the identity $\mathds{1}$. Thus $\mathrm{Alg}(\mathcal{S})$
is closed under multiplication and contains $\mathds{1}$.
\end{definition}
\begin{theorem}[Identifiability under minimal direction sets]
\label{thm:identifiability}
Let $H_\theta = \sum_{i=1}^m \theta_i h_i$ and $H' = \sum_{i=1}^m \theta_i^* h_i$ be Hamiltonians in the subspace $\mathcal{S} = \mathrm{span}\{h_i\}$. Suppose the direction operators $\{F_k\}_{k \in \mathcal{K}}$ satisfy
\begin{align}
    \mathrm{Comm}(\{F_k\}_{k \in \mathcal{K}}) \cap {\rm Alg}(\mathcal{S}) \subseteq \mathrm{span}\{\mathds{1}\}.
    \label{eq:suff_cond}
\end{align}
If $\mathscr{D}_k H_\theta = \mathscr{D}_k H'$ for all $k \in \mathcal{K}$, then $H_\theta - H' \in \mathrm{span}\{\mathds{1}\}$.
\end{theorem}
\begin{proof}
Suppose $\mathscr{D}_k H_\theta = \mathscr{D}_k H'$ for all $k$. Expanding the score map:
\begin{align}
    F_k - \e{-H_\theta} F_k \e{H_\theta} &= F_k - \e{-H'} F_k \e{H'} \\
    \implies \e{-H_\theta} F_k \e{H_\theta} &= \e{-H'} F_k \e{H'} \\
    \implies \e{H'}\e{-H_\theta} F_k &= F_k \,\e{H'}\e{-H_\theta}.
\end{align}
Define $Y := \e{H'}\e{-H_\theta}$. Then $Y \in \mathrm{Comm}(\{F_k\}_{k \in \mathcal{K}})$ and $Y \in \text{Alg}(\mathcal{S})$. Therefore:
\begin{align}
    Y \in \mathrm{Comm}(\{F_k\}_{k \in \mathcal{K}}) \cap \mathrm{Alg}(\mathcal{S}) \subseteq \mathrm{span}\{\mathds{1}\}.
\end{align}
Thus $H_\theta - H' = c \cdot \mathds{1}$ for some $c \in \mathbb{R}$.
\end{proof}
First, note that if the operators $\{F_k\}$ generate the full algebra, we have $\text{Comm}\{F_k\} = \text{span}(\mathds{1})$ and the condition in Eq.~\ref{eq:suff_cond} is always satisfied. Second, the ambiguity $H_\theta - H' = c \cdot \mathds{1}$ is physically inconsequential: constant shifts to the Hamiltonian do not affect the thermal state. Indeed,
\begin{align}
    \rho = \frac{\e{-H'}}{\Tr(\e{-H'})} = \frac{\e{-H_\theta + c\mathds{1}}}{\Tr(\e{-H_\theta + c\mathds{1}})} = \frac{\e{c}\cdot\e{-H_\theta}}{\e{c}\cdot\Tr(\e{-H_\theta})} = \rho_\theta\;,
\end{align}
and $H_\theta$ and $H'$ define the same density matrix.

\medskip

\subsection{Generalized Pauli (clock and shift) operators}
\label{sec:clock_shift}

In this appendix, we collect the basic properties of the clock and shift operators
used to construct single-site direction operators $F_k$ for qubit and lattice
systems (Sec.~\ref{sec:choose_Fk}). On a single site with finite local Hilbert-space
dimension $n$ and computational basis $\{|0\rangle,\dots,|n-1\rangle\}$, define
the \emph{clock} operator $\hat C$ and \emph{shift} operator $\hat S$ by
\begin{align}
    \hat C\,|k\rangle = \omega^{k}\,|k\rangle\;, \qquad
    \hat S\,|k\rangle = |k+1 \!\!\mod n\rangle\;, \qquad
    \omega := \e{2\pi\im/n}\;.
\end{align}
Equivalently, $\hat C = \mathrm{diag}(1,\omega,\dots,\omega^{n-1})$ and $\hat S$
is the cyclic permutation matrix. Both are unitary and obey
\begin{align}
    \hat S^n = \hat C^n = \mathds{1}\;, \qquad
    \hat S\hat C = \omega^{-1}\,\hat C\hat S\;.
    \label{eq:app_weyl}
\end{align}
The $n^2$ operators $\{\hat C^a \hat S^b\}_{a,b=0}^{n-1}$ form the
Heisenberg--Weyl basis of $M_n(\mathbb{C})$: they are linearly independent and
span the full single-site matrix algebra. Consequently, $\hat C$ and $\hat S$
generate $M_n(\mathbb{C})$ irreducibly, and the only operators commuting with
both are multiples of the identity,
\begin{align}
    \mathrm{Comm}\{\hat C,\hat S\} = \mathbb{C}\,\mathds{1}\;,
    \label{eq:app_comm_CS}
\end{align}
which is precisely the completeness property required for identifiability.

\medskip
\noindent\textbf{Qubit case ($n=2$).}
Here $\omega = \e{\im\pi} = -1$, and the clock and shift operators reduce to the
familiar Pauli matrices,
\begin{align}
    \hat C = \begin{pmatrix} 1 & 0 \\ 0 & -1 \end{pmatrix} = \sigma^z\;, \qquad
    \hat S = \begin{pmatrix} 0 & 1 \\ 1 & 0 \end{pmatrix} = \sigma^x\;.
\end{align}
The relation \eqref{eq:app_weyl} becomes the anticommutation relation
$\hat S\hat C = -\hat C\hat S$, i.e.\ $\sigma^x\sigma^z = -\sigma^z\sigma^x$, and
the third Pauli operator is the product $\sigma^y = \im\,\sigma^x\sigma^z = \im\,\hat S\hat C$
rather than an independent generator. Thus the single-site Pauli directions
$\{\sigma^z,\sigma^x\}$ used in the qubit examples
(Secs.~\ref{sec:phi4}~ff.) are exactly the $n=2$ instance of the clock--shift
construction, and the $n^2 = 4$ Heisenberg--Weyl elements
$\{\mathds{1},\sigma^x,\sigma^z,\sigma^x\sigma^z\}$ are the qubit analog of the
general qudit basis $\{\hat C^a\hat S^b\}$.

\medskip
\noindent\textbf{Continuum limit.}
Writing $\hat C = \e{\im a\,\hat p}$ and $\hat S = \e{-\im b\,\hat q}$ with real
$a,b$ satisfying $ab = 2\pi/n$, the Weyl relation \eqref{eq:app_weyl} is the
finite-dimensional analog of the canonical commutator $[\hat p,\hat q] = \im$.
As $n\to\infty$ with the lattice spacing scaled appropriately, $\hat p$ and
$\hat q$ converge to conjugate position and momentum operators and
\begin{align}
    \lim_{n\to\infty}[\hat p_{n,i},\hat q_{n,j}] = \im\,\delta_{ij}\;,
\end{align}
recovering the continuum directions $F_k \in \{\hat\phi_i,\hat\pi_i\}$ used for
the lattice $\phi^4$ theory in Sec.~\ref{sec:phi4}.\\

\medskip

\subsection{Fermionic systems}
\label{sec:fermi_majorana}
For fermions, the single-site clock--shift construction is replaced by the
single-mode Majorana operators, which play the same role for the
\emph{anticommuting} (Clifford) algebra that $\hat C,\hat S$ play for the
\emph{commuting} (Weyl) algebra of a distinguishable site. For $N$ fermionic
modes with annihilation and creation operators $c_k,c_k^\dagger$ obeying the
canonical anticommutation relations
$\{c_k,c_l^\dagger\}=\delta_{kl}$, $\{c_k,c_l\}=0$, define the $2N$ Majoranas
\begin{align}
    \gamma_{k,1} = c_k^\dagger + c_k\;, \qquad
    \gamma_{k,2} = \im(c_k^\dagger - c_k)\;, \qquad k = 1,\dots,N\;.
\end{align}
After relabeling by a single index $a=1,\dots,2N$, these operators are Hermitian and obey
the Clifford relation
\begin{align}
    \{\gamma_a,\gamma_b\} = 2\,\delta_{ab}\,\mathds{1}\;,
    \label{eq:app_clifford}
\end{align}
the anticommuting analog of the Weyl relation \eqref{eq:app_weyl}. The Fock
space has dimension $2^{N}$.

We claim that, exactly as for the clock--shift operators, the Majoranas generate
the full matrix algebra and have a trivial commutant,
\begin{align}
    \mathrm{Comm}\{\gamma_a\}_{a=1}^{2N} = \mathbb{C}\,\mathds{1}\;,
    \label{eq:app_comm_maj}
\end{align}
so that the completeness condition required for identifiability
[Eq.~\eqref{eq:ident_cond}] holds with no further assumptions. The argument is
the Clifford analog of the Heisenberg--Weyl basis statement above.

For an ordered index set $S=\{a_1<a_2<\dots<a_r\}\subseteq\{1,\dots,2N\}$, define
the Majorana monomial
\begin{align}
    \Gamma_S := \gamma_{a_1}\gamma_{a_2}\cdots\gamma_{a_r}\;,
    \qquad \Gamma_\emptyset := \mathds{1}\;,
\end{align}
which gives $2^{2N}$ operators in total (one for each subset $S$). From
\eqref{eq:app_clifford} we have $\gamma_a^2 = \mathds{1}$, and distinct
Majoranas anticommute. Thus, any product $\Gamma_S$ with
$S\neq\emptyset$ is traceless: $\mathrm{Tr}\,\Gamma_S = 0$. Second, the product
$\Gamma_S^\dagger\Gamma_{S'}$ is again a Majorana monomial (up to a sign), and it
equals $\mathds{1}$ exactly when $S=S'$ and is traceless otherwise. Taking the
trace therefore gives
\begin{align}
    \mathrm{Tr}\!\left(\Gamma_S^\dagger\,\Gamma_{S'}\right)
    = 2^{N}\,\delta_{S,S'}\;,
    \label{eq:app_maj_orthog}
\end{align}
so the monomials are orthogonal. Thus the $2^{2N}$ monomials $\{\Gamma_S\}$ are linearly independent. As
$\dim M_{2^N}(\mathbb{C}) = (2^N)^2 = 2^{2N}$, they form a \emph{basis} of the
full matrix algebra on Fock space, in exact parallel with the $n^2$
Heisenberg--Weyl elements $\{\hat C^a\hat S^b\}$ spanning $M_n(\mathbb{C})$.
Consequently, the Majoranas generate $M_{2^N}(\mathbb{C})$ irreducibly.
\\

To obtain \eqref{eq:app_comm_maj}, suppose $X$ commutes with every $\gamma_a$.
Using 
$[X,\gamma_a\gamma_b]=[X,\gamma_a]\gamma_b+\gamma_a[X,\gamma_b]$, we find that $X$
commutes with every monomial $\Gamma_S$, hence with all of
$M_{2^N}(\mathbb{C})$ by the basis property. An operator commuting with the
entire matrix algebra is a scalar (Schur's lemma), so $X = c\,\mathds{1}$. This
establishes \eqref{eq:app_comm_maj}.
\\

In particular, the $4L_s$ Majoranas of the $L_s$-site Fermi--Hubbard model
($N=2L_s$ modes from sites $\times$ spins) generate
$\mathrm{Cl}(4L_s,\mathbb{R})\cong M_{2^{2L_s}}(\mathbb{C})$ and satisfy
$\mathrm{Comm}\{\gamma_{k,\alpha}\}=\mathbb{C}\mathds{1}$, which is the
completeness statement used in Sec.~\ref{sec:choose_Fk}. 

The construction requires no Jordan--Wigner mapping.
To establish locality, however, we must also specify the parity
of the Hamiltonian terms.

\paragraph*{Locality of the nested commutators.}
For fermions, the support of an operator refers to the modes
whose creation and annihilation operators are needed to
express it. In particular, $\gamma_{k,\alpha}$ is supported
on mode $k$.

We assume that every term $h_i$ in
$H_\theta=\sum_i\theta_i h_i$ has even fermion parity:
each $h_i$ is a sum of products containing an even number
of creation and annihilation operators. Although individual
fermion operators on distinct modes anticommute, moving an
operator past an even number of such factors cancels the
minus signs. Consequently, an even Hamiltonian term commutes
with every operator on disjoint support:
\begin{align}
    [h_i,A]=0
    \qquad\text{if}\qquad
    \mathrm{supp}(h_i)\cap\mathrm{supp}(A)=\varnothing.
\end{align}
This also holds when $A$ is odd, as for our Majorana
directions~\cite{nachtergaele2018fermions}.

Thus, the first commutator is
\begin{align}
    [H_\theta,\gamma_{k,\alpha}]
    =\sum_{i:\,k\in\mathrm{supp}(h_i)}
      \theta_i[h_i,\gamma_{k,\alpha}],
\end{align}
so only Hamiltonian terms involving mode $k$ contribute.
Each resulting term is supported on the modes appearing
in that Hamiltonian term. At the next commutator, only
Hamiltonian terms overlapping this new support can
contribute. The same argument applies at every subsequent
order.

If the interaction range is $R$, meaning that no two sites
in the support of a local Hamiltonian term are separated
by more than $R$, each commutator can extend the support
by at most this distance. Therefore,
\begin{align}
    \mathrm{supp}\bigl([H_\theta,\gamma_{k,\alpha}]_\ell\bigr)
    \subseteq \{j:\mathrm{dist}(j,k)\leq\ell R\}.
\end{align}
Here, distance is measured between the sites carrying the
modes; modes on the same site have zero spatial separation.
This gives the same bounded support growth as in the spin case.
The hopping, density, and interaction terms of the
Fermi--Hubbard model all satisfy the even-parity assumption.

If odd Hamiltonian terms are allowed, this argument can fail.
For example, take $H=\lambda\gamma_{j,1}$ with real
$\lambda\neq0$. For any distinct mode $k$,
\begin{align}
    [H,\gamma_{k,1}]
    =2\lambda\gamma_{j,1}\gamma_{k,1}\neq0,
\end{align}
regardless of the distance between $j$ and $k$.
Thus, bounded spatial support growth follows from the
locality and even parity of the Hamiltonian terms;
working directly in Fock space alone is insufficient.

\medskip

\section{Explicit example of nonconvexity of the loss}
\label{app:nonconvexity_example}

The completeness condition on the direction operators ensures that the true Hamiltonian is the unique global minimizer, up to an additive constant. However, does this also imply that the loss is easy to minimize from any initial point? We show with a two-qubit example that the answer is no: even when the direction operators generate the full matrix algebra, the loss can have a local minimum at which the Hamiltonian parameters differ from their true values.

Consider the parametrized Hamiltonian
\begin{align}
    H_\theta &= \theta_1 h_1+\theta_2 h_2, \\
    h_1 &= 0.19\,I\otimes Z-0.81\,X\otimes Z, \\
    h_2 &= 0.73\,I\otimes Z+0.17\,X\otimes Z
    -0.21\,Y\otimes I,
    \label{eq:nonconvex_hamiltonian}
\end{align}
where $X,Y,Z$ are the Pauli matrices and $I$ is the single-qubit identity. We choose the target parameters to be
\begin{align}
    \theta^*=(-4,-3.5),\qquad
    \rho=\frac{e^{-H_{\theta^*}}}{\Tr(e^{-H_{\theta^*}})},
\end{align}
with the inverse temperature absorbed into the Hamiltonian, as in the main text. For the direction operators, we take
\begin{align}
    \{F_k\}_{k=1}^4
    =\{X\otimes I,Z\otimes I,I\otimes X,I\otimes Z\}.
\end{align}
These operators generate $M_4(\mathbb C)$, and their common commutant therefore contains only multiples of the identity. Moreover, $h_1$ and $h_2$ are linearly independent and traceless, so there is no additive-constant ambiguity within this parametrization. Thus, the loss vanishes only at $\theta=\theta^*$.

To examine the loss away from this point, we evaluate the matrix exponentials directly, without truncating the nested-commutator expansion. Writing
\begin{align}
    R_k(\theta)&:=e^{-H_\theta}F_k e^{H_\theta}, \\
    \Delta_k(\theta)&:=R_k(\theta)-R_k(\theta^*),
\end{align}
the loss in Eq.~\eqref{eq:Lstd_def} becomes
\begin{align}
    \mathcal L_{\rm osm}(\theta)
    =\sum_{k=1}^4\Tr\!\left(
    \rho\,\Delta_k(\theta)^\dagger\Delta_k(\theta)
    \right).
    \label{eq:nonconvex_exact_loss}
\end{align}
The sign difference between $\Delta_k$ and the score residual does not affect this expression, since $\mathscr D_kH_\theta=F_k-R_k(\theta)$.

Solving the stationarity equations numerically gives a second minimum at
\begin{align}
    \theta_{\rm m}&\simeq(-1.44731998,-5.28849816), \\
    \mathcal L_{\rm osm}(\theta_{\rm m})
    &\simeq 1.598554732\times10^5.
\end{align}
At the unrounded numerical solution, the gradient norm is below $10^{-8}$. The Hessian eigenvalues are
\begin{align}
    \lambda\!\left(\nabla^2\mathcal L_{\rm osm}(\theta_{\rm m})\right)
    \simeq\{2.17745\times10^4,\,3.51735\times10^5\}.
\end{align}
Both eigenvalues are positive, providing numerical evidence for a strict local minimum. Since the loss is positive here, while $\mathcal L_{\rm osm}(\theta^*)=0$, this minimum is spurious. The same example also has a saddle point at
\begin{align}
    \theta_{\rm s}&\simeq(-1.58465053,-4.93358407), \\
    \mathcal L_{\rm osm}(\theta_{\rm s})
    &\simeq1.603146581\times10^5,
\end{align}
with Hessian eigenvalues approximately $-1.56466\times10^4$ and $1.29538\times10^5$. We computed the gradients using Fr\'echet derivatives of the matrix exponential and checked the Hessian eigenvalues by central differences of these gradients with step sizes from $10^{-4}$ to $10^{-6}$.

This example separates two properties of the learning problem. Completeness ensures that matching all the operator scores determines the true Hamiltonian, but it does not rule out stationary points with nonzero score residuals. In particular, a positive-loss stationary point is incompatible with a global PL inequality with a strictly positive constant. This does not contradict the local PL inequality established in Appendix~\ref{app:pl}, which applies in a neighborhood of the true parameters. The obstruction in this example is already present in the exact population loss, and therefore does not arise from finite-sample errors or truncation of the score expansion.

\medskip

\medskip

\section{Local PL inequality for the standard-metric loss}
\label{app:pl}
This appendix provides a self-contained proof that $\mathcal{L}_{\mathrm{std}}$ satisfies a
\emph{local} Polyak--\L{}ojasiewicz (PL) inequality near a realizable minimizer under natural
algebraic assumptions. We assume that the Hamiltonian is linear in the parameters
\begin{align}
    H_\theta &= \sum_{i=1}^m \theta_i h_i .
\end{align}

\medskip

\subsection{Real Hilbert structure and residual map}

Let $\mathcal{H}:=\mathbb{C}^{d\times d}$ and regard it as a real Hilbert space equipped with the
$\rho$-weighted inner product
\begin{align}
    (A,B)_{\rho} := \Re\,\Tr(\rho\,A^\dagger B),\qquad
    \|A\|_{\rho}:=\sqrt{(A,A)_{\rho}},
\end{align}
where $\rho=\e{-H'}/Z'$ is the data state and is full rank.
For $n$ directions, define the product space $\mathcal{H}^n$ with
\begin{align}
    (r,s)_{\rho,n} := \sum_{k=1}^n (r_k,s_k)_{\rho},
    \qquad
    \|r\|_{\rho,n}^2=(r,r)_{\rho,n}.
\end{align}

Fix a realizable target $H':=H_{\theta^\ast}$ and define the residual
\begin{align}
    r(\theta) := (\Delta_1(\theta),\dots,\Delta_n(\theta))\in\mathcal{H}^n,
    \qquad
    \mathcal{L}_{\mathrm{std}}(\theta):=\|r(\theta)\|_{\rho,n}^2 .
\end{align}
Then $\mathcal{L}_{\mathrm{std}}(\theta^\ast)=0$.

\medskip

\subsection{Jacobian, adjoint, and gradient formula}
Assume $r$ is Fr\'echet differentiable. Define the Jacobian
$J(\theta):\mathbb{R}^m\to\mathcal{H}^n$ by
\begin{align}
    (J(\theta)u)_k &:= \sum_{i=1}^m u_i\,\partial_{\theta_i}\Delta_k(\theta),\notag\\
    J(\theta)u &:= Dr(\theta)[u].
\end{align}
Let $J(\theta)^\top:\mathcal{H}^n\to\mathbb{R}^m$ denote the real adjoint with respect to the
inner product $(\cdot,\cdot)_{\rho,n}$, i.e.
\begin{align}
    (J(\theta)u,y)_{\rho,n} = (u,J(\theta)^\top y)_{\mathbb{R}^m}.
\end{align}

The directional derivative of $\mathcal{L}_{\mathrm{std}}$ along $u\in\mathbb{R}^m$ is
\begin{align}
    \frac{d}{dt}\bigg|_{t=0}\mathcal{L}_{\mathrm{std}}(\theta+tu)
    &= 2\left(r(\theta),\,\frac{d}{dt}\bigg|_{t=0}r(\theta+tu)\right)_{\rho,n} \notag\\
    &= 2(J(\theta)u,r(\theta))_{\rho,n} \notag\\
    &= 2(J(\theta)^\top r(\theta),u)_{\mathbb{R}^m}.
\end{align}
Therefore,
\begin{align}
    \nabla\mathcal{L}_{\mathrm{std}}(\theta)
    = 2\,J(\theta)^\top r(\theta).
    \label{eq:app_grad_rho}
\end{align}

\medskip

\subsection{Derivative of the score map}

We use the Fr\'echet derivative identity for the matrix exponential:
\begin{align}
    \frac{d}{dx}\e{C(x)}
    = \int_0^1 \e{(1-s)C(x)}\,C'(x)\,\e{sC(x)}\,ds .
\end{align}
Define the time-averaged conjugation of $h_i$ by $H_\theta$:
\begin{align}
    v_i(\theta) := \int_0^1 \e{sH_\theta}\,h_i\,\e{-sH_\theta}\,ds .
\end{align}
Then we have
\begin{align}
    \partial_{\theta_i}\Delta_k(\theta) &= \partial_{\theta_i}(\mathscr{D}_kH_\theta) = \partial_{\theta_i}([F_k,\e{-H_\theta}]\e{H_\theta})\notag \\
    &= -\partial_{\theta_i}(\e{-H_\theta}F_k\,\e{H_\theta})\notag \\
    &= - \partial_{\theta_i}(\e{-H_\theta})F_k\,\e{H_\theta}-\e{-H_\theta}F_k\,\partial_{\theta_i}(\e{H_\theta})\;.
    \label{eq:di_deltak}
\end{align}
We now use
\begin{align}
     \partial_{\theta_i}(\e{-H_\theta}) &= -\int_0^1 \e{-(1-s)H_\theta}(\partial_{\theta_i}H_\theta)\e{-sH_\theta}\,ds\notag \\
     &= -\int_0^1 \e{-(1-s)H_\theta}h_i\e{-sH_\theta}\,ds\notag \\
     &= -\e{-H_\theta}v_i(\theta)\;.
\end{align}
Similarly, using $\partial_{\theta_i}(\e{H_\theta}\e{-H_\theta}) = 0$, it is easy to show that
\begin{align}
    \partial_{\theta_i}(\e{H_\theta}) &= v_i(\theta)\e{H_\theta}\;.
\end{align}
Substituting these relations into Eq.~\eqref{eq:di_deltak}, we get
\begin{align}
    \partial_{\theta_i}\Delta_k(\theta)
    = \e{-H_\theta}\,[v_i(\theta),F_k]\,\e{H_\theta}.
    \label{eq:app_dDelta_rho}
\end{align}

\medskip

\subsection{Gram matrix and uniform positive definiteness}
\label{sec:gm_upd}

Define the Gram matrix $G(\theta)\in\mathbb{R}^{m\times m}$ by
\begin{align}
    G_{ij}(\theta)
    &:= (\partial_{\theta_i}r(\theta),\partial_{\theta_j}r(\theta))_{\rho,n} \notag\\
    &= \sum_{k=1}^n
       \Re\,\Tr\!\left(\rho\,
       (\partial_{\theta_i}\Delta_k(\theta))^\dagger
       \partial_{\theta_j}\Delta_k(\theta)\right).
    \label{eq:app_Gram_rho}
\end{align}
Then for all $u\in\mathbb{R}^m$,
\begin{align}
    \|J(\theta)u\|_{\rho,n}^2 = u^\top G(\theta)u,
    \qquad
    \sigma_{\min}(J(\theta))^2 = \lambda_{\min}(G(\theta)).
\end{align}

Assume:
\begin{enumerate}[label=(S\arabic*), ref=S\arabic*]
    \item \label{assump:S1} The operators $\{h_i\}_{i=1}^m$ are linearly independent over $\mathbb{R}$.
    \item \label{assump:S2} The operators $\{h_i\}_{i=1}^m$ are traceless.
    \item \label{assump:S3} The operators $\{F_k\}$ generate the full matrix algebra, so that
    $[X,F_k]=0\ \forall k \Rightarrow X=c\,\mathds{1}$.
    \item \label{assump:S4} $\Omega\subset\mathbb{R}^m$ is compact and contains $\theta^\ast$.
\end{enumerate}

\paragraph*{Uniform Gram bound.}
Under~\cref{assump:S1,assump:S2,assump:S3,assump:S4} and $\rho\succ0$, one has $G(\theta)\succ0$ for all $\theta\in\Omega$ and
\begin{align}
    \mathfrak{a}^2 := \inf_{\theta\in\Omega}\lambda_{\min}(G(\theta)) > 0 .
\end{align}
\begin{proof}
Fix $\theta\in\Omega$ and suppose $u^\top G(\theta)u=0$.
Then $\|J(\theta)u\|_{\rho,n}=0$, hence each component
$\sum_i u_i\,\partial_{\theta_i}\Delta_k(\theta)$ vanishes as an operator because
$\rho\succ0$.
Using \eqref{eq:app_dDelta_rho} and the invertibility of conjugation by $\e{\pm H_\theta}$ gives
\begin{align}
    \left[\sum_{i=1}^m u_i v_i(\theta),\,F_k\right] = 0 \qquad \forall k.
\end{align}
By \cref{assump:S3}, $\sum_i u_i v_i(\theta)=c\,\mathds{1}$ for some $c\in\mathbb{R}$.
Let $X_u:=\sum_i u_i h_i$. Then
\begin{align}
    \sum_i u_i v_i(\theta)
    = \int_0^1 \e{sH_\theta}X_u \e{-sH_\theta}\,ds =c\cdot\mathds{1}.
    \label{eq:Xu_c1}
\end{align}

Since the operators $h_i$ are traceless, $X_u$ is also traceless. Moreover, $\text{Tr}[ \int_0^1 \e{sH_\theta}X_u \e{-sH_\theta}\,ds] =  \int_0^1  \text{Tr}[\e{sH_\theta}X_u \e{-sH_\theta}\,] ~ ds = 0$. Thus $c = 0$ and   $\int_0^1 \e{sH_\theta}X_u \e{-sH_\theta}\,ds = 0$. 
Now (for fixed $\theta$) consider this quantity in an eigenbasis of $H_\theta$, $H_\theta|m\rangle  = E_m|m\rangle$:
\begin{align}
    \left(\int_0^1 \e{sH_\theta}X_u \e{-sH_\theta}\,ds\right)_{mn} &= (X_u)_{mn}\int_0^1\e{s(E_m-E_n)}\,ds\notag \\
    &= (X_u)_{mn}\cdot g(E_m-E_n)\;,
\end{align}
where we define 
\begin{align}
    g(x) := \frac{\e{x}-1}{x}\;.
\end{align}
Note that $g(x) \neq 0$ for all finite $x$. Hence, every matrix element of $X_u$ in this basis vanishes, implying that $X_u = 0$. Since we assumed the linear independence of the operators $h_i$, this in turn implies that $u = 0$.\\

This shows that $u^\top G (\theta) u = 0 \ \implies \ u = 0 $. We know from the definition of $G$ that $u^\top G (\theta) u$ is non-negative and therefore $u \neq 0 \implies u G u^\top > 0$.


Thus $G(\theta)\succ0$ for all $\theta$.
Continuity of $G$ and the compactness assumption \cref{assump:S4} imply a uniform lower bound.

\end{proof}

\medskip

\subsection{Projection--singular value inequality}
Fix $\theta$ and abbreviate $J := J(\theta)$.
We view
\begin{align}
J : \mathbb{R}^m \longrightarrow \mathcal{H}^n
\end{align}
as a bounded linear operator between Hilbert spaces
\begin{align}
(\mathbb{R}^m,\langle\cdot,\cdot\rangle)
\quad \text{and} \quad
(\mathcal{H}^n,(\cdot,\cdot)_{\rho,n}),
\end{align}
where $\langle\cdot,\cdot\rangle$ is the Euclidean inner product.
Let
\begin{align}
J^\top : \mathcal{H}^n \longrightarrow \mathbb{R}^m
\end{align}
denote the real Hilbert adjoint defined by
\begin{align}
(Ju,y)_{\rho,n} = \langle u, J^\top y\rangle,
\qquad
\forall\,u\in\mathbb{R}^m,\ y\in\mathcal{H}^n .
\end{align}
Let
\begin{align}
\Pi := \Pi_{\mathrm{Range}(J)}
\end{align}
denote the projection onto $\mathrm{Range}(J)$ in
$(\mathcal{H}^n,(\cdot,\cdot)_{\rho,n})$.

Define
\begin{align}
\sigma_{\min}(J)
&:= \inf_{\lVert u\rVert=1} \lVert Ju\rVert_{\rho,n} := \mathfrak{a} \\
\sigma_{\max}(J)
&:= \sup_{\lVert u\rVert=1} \lVert Ju\rVert_{\rho,n} := \mathfrak{b} .
\end{align}
Equivalently,
\begin{align}
\sigma_{\min}(J)^2 = \lambda_{\min}(J^\top J) = \mathfrak{a}^2 ,
\end{align}
where $J^\top J : \mathbb{R}^m \to \mathbb{R}^m$ is self-adjoint and positive semidefinite.

\begin{lemma}[Projection--singular value inequality]
For every $y\in\mathcal{H}^n$,
\begin{align}
\lVert J^\top y\rVert
\;\ge\;
\sigma_{\min}(J)\,
\lVert \Pi y\rVert_{\rho,n}.
\label{eq:proj_sv_weighted}
\end{align}
\end{lemma}
\begin{proof}
    Decompose $y = y_{\perp} + y_{\parallel}$, where the perpendicular and parallel components are with respect to $J$. Now note that 
    \begin{align}
        J^\top y &= J^\top(y_\perp + y_\parallel)\notag \\
        &= (J,y_\perp)_{\rho,n} + (J,y_\parallel)_{\rho,n} \notag \\
        &= (J,y_\parallel)_{\rho,n}\;.
    \end{align}

Since $y_\parallel \in \text{Range}(J)$ by definition, $\|J^T y \| = \| J^T y_\parallel  \| \geq \sigma_{min} (J) \| \Pi y \|_{\rho,n}$. 

\end{proof}


\subsection{Polyak-\L{}ojasiewicz Inequality}
We establish a \emph{local} PL inequality for the score-matching loss, which guarantees convergence of gradient descent to the global minimum.

\subsubsection{Fundamental constants}

The following constants characterize the geometry of the optimization landscape:

\begin{definition}[Gram matrix bounds]
\label{def:gram_bounds}
For a compact parameter domain $\Omega \ni \theta^*$, define:
\begin{align}
\mathfrak{a}&:= \inf_{\theta \in \Omega} \sigma_{\min}(J(\theta)) > 0 \\
\mathfrak{b} &:= \sup_{\theta \in \Omega} \sigma_{\max}(J(\theta)) < \infty
\end{align}
where $\sigma_{\min}$ and $\sigma_{\max}$ denote the minimum and maximum singular values. Equivalently, $\mathfrak{a}^2 = \inf_\theta \lambda_{\min}(G(\theta))$ and $\mathfrak{b}^2 = \sup_\theta \lambda_{\max}(G(\theta))$, where $G(\theta) = J(\theta)^\top J(\theta)$ is the Gram matrix.
\end{definition}

The bounds $\mathfrak{a}$ and $\mathfrak{b}$ quantify identifiability: $\mathfrak{a} > 0$ ensures that distinct parameters produce distinct scores (the inverse problem is well-posed), while $\mathfrak{b} < \infty$ ensures stability.

\begin{definition}[Lipschitz constant]
\label{def:lipschitz}
The residual $r(\theta)$ is Lipschitz continuous with constant $L$ if
\begin{align}
\|r(\theta_1) - r(\theta_2)\|_{\rho,n} \leq L \|\theta_1 - \theta_2\|
\end{align}
for all $\theta_1, \theta_2 \in \Omega$. Equivalently, $\|J(\theta)\|_{\mathrm{op}} \leq L$ for all $\theta$.
\end{definition}

\begin{definition}[Jacobian Lipschitz constant]
\label{def:jacobian_lipschitz}
The Jacobian $J(\theta)$ is Lipschitz continuous with constant $L_J$ if
\begin{align}
\|J(\theta_1) - J(\theta_2)\|_{\mathrm{op}} \leq L_J \|\theta_1 - \theta_2\|
\end{align}
for all $\theta_1, \theta_2 \in \Omega$.
\end{definition}

\subsubsection{Local range condition}

The key technical result is that near the true parameters, the residual $r(\theta)$ remains well-aligned with the range of the Jacobian $J(\theta)$.

\begin{theorem}[Local range condition]
\label{thm:local_range}
Let $\Pi_\theta := \Pi_{\mathrm{Range}(J(\theta))}$ denote the orthogonal projection onto the range of $J(\theta)$. For
\begin{align}
\delta := \frac{\mathfrak{a}}{3L_J}
\end{align}
and all $\theta \in B(\theta^*, \delta)$, the residual satisfies
\begin{align}
\|\Pi_\theta \, r(\theta)\|_{\rho,n} \geq \gamma \|r(\theta)\|_{\rho,n}
\label{eq:local_range_cond}
\end{align}
with
\begin{align}
\gamma := \frac{\mathfrak{a}}{2\mathfrak{b} + \mathfrak{a}}.
\label{eq:gamma_def}
\end{align}
\end{theorem}

\begin{proof}
Let $u = \theta - \theta^*$. Since $r(\theta^*) = 0$, Taylor expansion gives
\begin{align}
r(\theta) = J(\theta^*) u + R(\theta, \theta^*)
\end{align}
where the remainder satisfies $\|R(\theta, \theta^*)\|_{\rho,n} \leq \frac{L_J}{2}\|u\|^2$ by the Lipschitz continuity of $J$.

Rewriting in terms of $J(\theta)$ gives
\begin{align}
r(\theta) &= J(\theta) u + \underbrace{(J(\theta^*) - J(\theta))u + R(\theta, \theta^*)}_{=: \tilde{R}(\theta, \theta^*)}.
\end{align}

The modified remainder satisfies
\begin{align}
\|\tilde{R}(\theta, \theta^*)\|_{\rho,n} 
&\leq \|J(\theta^*) - J(\theta)\|_{\mathrm{op}} \|u\| + \frac{L_J}{2}\|u\|^2 \notag\\
&\leq L_J \|u\|^2 + \frac{L_J}{2}\|u\|^2 = \frac{3L_J}{2}\|u\|^2.
\end{align}

Projecting onto $\mathrm{Range}(J(\theta))$ and using $\Pi_\theta J(\theta) = J(\theta)$ gives
\begin{align}
\Pi_\theta \, r(\theta) = J(\theta) u + \Pi_\theta \tilde{R}(\theta, \theta^*).
\end{align}

By the triangle inequality and $\sigma_{\min}(J(\theta)) \geq \alpha$, we obtain
\begin{align}
\|\Pi_\theta \, r(\theta)\|_{\rho,n} 
&\geq \|J(\theta) u\|_{\rho,n} - \|\tilde{R}(\theta, \theta^*)\|_{\rho,n} \notag\\
&\geq \mathfrak{a} \|u\| - \frac{3L_J}{2}\|u\|^2.
\label{eq:Pr_lower}
\end{align}

For $\|u\| \leq \delta = \mathfrak{a}/3L_J$, we have
\begin{align}
\|\Pi_\theta \, r(\theta)\|_{\rho,n} \geq \frac{\mathfrak{a}}{2}\|u\|.
\label{eq:Pr_lower_final}
\end{align}

For the upper bound on $\|r(\theta)\|_{\rho,n}$, we have
\begin{align}
\|r(\theta)\|_{\rho,n} 
&\leq \|J(\theta) u\|_{\rho,n} + \|\tilde{R}(\theta, \theta^*)\|_{\rho,n} \notag\\
&\leq \mathfrak{b} \|u\| + \frac{3L_J}{2}\|u\|^2 \notag\\
&\leq \left(\mathfrak{b} + \frac{\mathfrak{a}}{2}\right)\|u\|
\label{eq:r_upper}
\end{align}
where the last inequality uses $\|u\| \leq \frac{\mathfrak{a}}{3L_J}$.

Combining \eqref{eq:Pr_lower_final} and \eqref{eq:r_upper} gives
\begin{align}
\|\Pi_\theta \, r(\theta)\|_{\rho,n} \geq \frac{\mathfrak{a}/2}{\mathfrak{b} + \mathfrak{a}/2} \|r(\theta)\|_{\rho,n} = \gamma \|r(\theta)\|_{\rho,n}.
\end{align}
\end{proof}

\subsubsection{PL inequality}

\begin{theorem}[Local Polyak-\L{}ojasiewicz inequality]
\label{thm:pl_inequality}
For $\theta \in B(\theta^*, \delta)$ with $\delta = \mathfrak{a}/3L_J$, the score-matching loss satisfies
\begin{align}
\|\nabla \mathcal{L}(\theta)\|^2 \geq \mu \cdot \mathcal{L}(\theta)
\label{eq:pl_inequality}
\end{align}
with PL constant
\begin{align}
\mu = 4\mathfrak{a}^2 \gamma^2 = \frac{\mathfrak{a}^4}{(\mathfrak{b} + \mathfrak{a}/2)^2}\;.
\label{eq:mu_def}
\end{align}
\end{theorem}

\begin{proof}
The gradient of the loss is
\begin{align}
\nabla \mathcal{L}(\theta) = 2 J(\theta)^\top r(\theta).
\end{align}

For any matrix $A$ and vector $v$, we have $\|A^\top v\| \geq \sigma_{\min}(A) \|\Pi_{\mathrm{Range}(A)} v\|$. Applying this with $A = J(\theta)$ and $v = r(\theta)$ gives
\begin{align}
\|\nabla \mathcal{L}(\theta)\| 
&= 2\|J(\theta)^\top r(\theta)\| \notag\\
&\geq 2 \sigma_{\min}(J(\theta)) \|P_\theta \, r(\theta)\|_{\rho,n} \notag\\
&\geq 2\mathfrak{a} \cdot \gamma \|r(\theta)\|_{\rho,n}
\end{align}
where we used $\sigma_{\min}(J(\theta)) \geq \mathfrak{a}$ and the local range condition \eqref{eq:local_range_cond}.

Squaring both sides and using $\mathcal{L}(\theta) = \|r(\theta)\|_{\rho,n}^2$ gives
\begin{align}
\|\nabla \mathcal{L}(\theta)\|^2 \geq 4\mathfrak{a}^2 \gamma^2 \mathcal{L}(\theta) = \mu \cdot \mathcal{L}(\theta).
\end{align}
\end{proof}

\paragraph*{Convergence guarantee.}
For gradient descent with step size $\eta = \Theta( 1/L)$, where $L$ is the Lipschitz constant of $\nabla \mathcal{L}$, the PL inequality implies linear convergence:
\begin{align}
\mathcal{L}(\theta_k) - \mathcal{L}(\theta^*) \leq (1 - \eta \mu)^k \left(\mathcal{L}(\theta_0) - \mathcal{L}(\theta^*)\right).
\end{align}
The number of iterations to reach $\epsilon$-accuracy is $O(\mu^{-1} \log(1/\epsilon))$ \cite{karimi2016}.
 ~~

\end{document}